\documentclass[pdflatex,sn-mathphys-num]{sn-jnl}

\usepackage{graphicx}%
\usepackage{multirow}%
\usepackage{amsmath,amssymb,amsfonts}%
\usepackage{amsthm}%
\usepackage{mathrsfs}%
\usepackage[title]{appendix}%
\usepackage{xcolor}%
\usepackage{textcomp}%
\usepackage{manyfoot}%
\usepackage{booktabs}%
\let\orcidlogo\undefined%
\usepackage{orcidlink}%
\usepackage{placeins}%
\usepackage{fullpage}%

\makeatletter
\renewenvironment{table}[1][]{%
  \begin{tableorg}[#1]%
  \tablebodyfont%
  \begin{center}%
}{\end{center}\end{tableorg}}

\renewcommand\paragraph{\@startsection{paragraph}{4}{\z@}%
  {-12pt \@plus -4pt \@minus-2pt}%
  {-1em}%
  {\reset@font\fontsize{10bp}{12bp}\bfseries\itshape\selectfont}}
\makeatother

\theoremstyle{thmstyleone}%
\newtheorem{theorem}{Theorem}%
\newtheorem{proposition}[theorem]{Proposition}%
\newtheorem{corollary}[theorem]{Corollary}%

\theoremstyle{thmstyletwo}%
\newtheorem{remark}{Remark}%

\theoremstyle{thmstylethree}%
\newtheorem{definition}{Definition}%

\newcommand{\Vt}{{V_{t}}}
\newcommand{\psit}{{\Psi_{t}}}
\newcommand{\Hcons}{\hat{\mathcal{H}}^{\mathrm{c}}_{t}}
\newcommand{\Kop}{\hat{K}}
\newcommand{\prare}{{p_{\mathrm{rare}}}}

\newcommand{\ket}[1]{\left|#1\right\rangle}

\hypersetup{bookmarksdepth=3}

\definecolor{darkgreen}{RGB}{0,100,0}

\begin{document}

\title[Quantum simulation of generative models]{Scalable quantum simulation of continuous-time generative models via tensor networks}

\author*[1]{\fnm{Nathan X.} \sur{Kodama}\,\orcidlink{0000-0002-6818-9669}}\email{nathan@sygaldry.com}
\author[1]{\fnm{L. Andrew} \sur{Wray}\,\orcidlink{0000-0001-8977-0201}}
\author[1,2]{\fnm{Sam} \sur{Cochran}\,\orcidlink{0000-0002-5538-641X}}
\equalcont{All work performed at Sygaldry Technologies, Inc.}
\author[1]{\fnm{Chad} \sur{Rigetti}}
\author[1,2]{\fnm{Shravan} \sur{Veerapaneni}\,\orcidlink{0000-0002-2294-7233}}
\equalcont{All work performed at Sygaldry Technologies, Inc.}
\author[1]{\fnm{Michael J.} \sur{Keiser}\,\orcidlink{0000-0002-1240-2192}}

\affil[1]{\orgname{Sygaldry Technologies, Inc.},
\orgaddress{\city{Ann Arbor}, \state{MI}, \country{USA}}}
\affil[2]{\orgname{University of Michigan},
\orgaddress{\city{Ann Arbor}, \state{MI}, \country{USA}}}

\abstract{%
\unboldmath 
Continuous-time flow and diffusion models are widely used across
many application domains, from large-scale deployment in computer vision and protein
folding to emerging adoption for modeling language, time series, and
quantum states.
After training, inferring statistical properties from continuous-time models
is costly. Wavefunction flows target this cost by recasting learned
transport as unitary evolution, whose final Born distribution approximates the
target distribution. This prepares a coherent amplitude encoding (a \emph{qsample})
that can be post-processed by quantum algorithms offering a quadratic advantage
over Monte Carlo sampling. We present the first
numerical study of these flows, in which we represent time-dependent potentials and states as tensor networks. At spatial dimension $d{=}8$, storage falls
by $\sim\!10^7\times$ relative to the dense grid of $N^d$ points, and evolution
wall-clock time falls
by $\gtrsim\!10^3\times$ against a baseline extrapolated from the measured
$d\le5$ scaling. We validate our pipeline by reproducing the
$\mathcal{O}(1/\sqrt{\prare})$ scaling of rare-event sampling.
}


\maketitle

\begin{figure}[tbp]
\centering
\includegraphics[width=\textwidth]{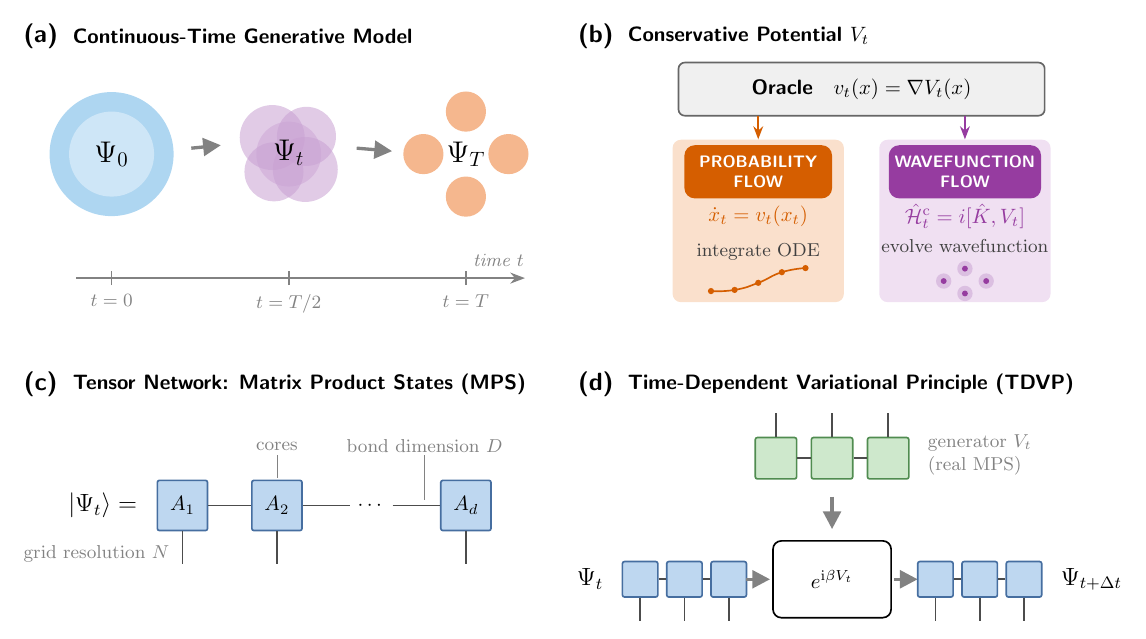}
\caption{\textbf{Wavefunction flows and their tensor network
realization.}
(a)~A Gaussian source wavefunction $\Psi_0$ is evolved through an
intermediate state $\Psi_t$ into a target $\Psi_T$ via the
Trotterized unitary $U(T) = \mathcal{T}\exp(-i\int_0^T \Hcons\,dt)$,
$\Hcons = i[\Kop,\Vt]$ (time-ordered because $\Vt$, hence $\Hcons$, is
time-dependent), along the generation-time axis $t$. Each Trotter
step alternates the kinetic step $\Kop$ (separable, applied via
one-dimensional Fourier transforms) with the position-diagonal potential step
$e^{i\beta\Vt}$. The final Born distribution supplies the coherent amplitude
representation required by amplitude amplification algorithms
when a suitable controlled preparation circuit is available.
(b)~A single velocity potential $V_t$, learned once by joint action
matching, drives two realizations of the same transport. The
\emph{probability flow} realization integrates the ordinary differential
equation (ODE)
$\dot{x}_t = v_t(x_t)$, $v_t = \nabla V_t$, to transport samples;
the \emph{wavefunction flow} realization evolves the wavefunction under
$\Hcons = i[\Kop,\Vt]$ with tensor network compression and samples
$|\Psi_T|^2$. This work develops and benchmarks a classical tensor network
simulation of the wavefunction realization.
(c)~State encoding for the flow wavefunction $\psit$: a matrix product
state (MPS) with
one rank-$\le D$ core per spatial axis (physical dimension $N$),
so the network has one site per spatial dimension $d$ and memory
$\mathcal{O}(d N D^2)$ in place of the $N^d$ dense grid.
(d)~The position-diagonal potential step $e^{i\beta\Vt}$ is applied
from the real generator $\Vt$ rather than as a generic matrix product
operator (MPO). Tensor cross interpolation (TCI) interpolates $\Vt$, and a time-dependent variational principle (TDVP) step evolves $\psit$ under it
(the two-site variant 2TDVP, and the one-site variant with subspace
expansion 1TDVP+SE, grow the bond adaptively; one-site 1TDVP alone is
fixed-rank);
a trained MPS $\Vt$ supplies the cores directly and removes the runtime TCI construction.}
\label{fig:overview}
\end{figure}

\section{Introduction}\label{sec:intro}

Continuous-time generative models including flow
matching~\cite{bib:flowmatching}, neural ordinary differential
equations~\cite{bib:chen2018}, and score-based
diffusion~\cite{bib:song2021scoresde} are widely used for modeling complex probability distributions, and because they
transport probability through physically motivated dynamics, they map
naturally onto physical hardware, including quantum processors. Once
trained, the model's statistical properties are costly to extract: any expectation, moment,
or property test must be estimated by averaging, so each added digit of precision
costs a hundredfold more samples. Wavefunction flows~\cite{bib:layden2025} target
this cost by recasting the classical continuity equation
transport that drives a learned source-to-target flow as
the squared modulus of a wavefunction evolving under a Hamiltonian
fixed by the generative model. The final state is a coherent
amplitude encoding, a \emph{qsample}~\cite{bib:aharonov2003} 
of the target distribution (Figure~\ref{fig:overview}). Coherent processing of qsamples yields a quadratic advantage over Monte Carlo sampling: amplitude
estimation measures a probability at
linear-in-precision rather than inverse-square cost~\cite{bib:bhmt,bib:montanaro2015}, and amplitude
amplification draws rare samples at $\mathcal{O}(1/\sqrt{\prare})$ rather
than the $\mathcal{O}(1/\prare)$ of rejection sampling.

Despite this promise, wavefunction flows have so far been a purely
theoretical proposal, unexplored beyond mathematical analysis,
because a direct numerical study is obstructed by the exponential cost of
representing the wavefunction on a dense grid. An eight-factor commutator product formula for the unitary (local
error $\mathcal{O}(\Delta t^2)$, first-order global convergence)
alternates a separable kinetic step on a periodic position grid, applied site-locally via one-dimensional
Fourier transforms, with a position-diagonal potential ($V$)
step that is cheap pointwise but whose exponentiation and dense-grid storage scale exponentially in the spatial dimension. At
a modest per-axis grid of $N{=}32$, the wavefunction already
comprises roughly $10^9$ amplitudes in six dimensions. The
$V$-step is therefore the dominant cost of classically simulating the
unitary; a reversible gate-level implementation would face the same step as
its bottleneck.

Motivated by the roles tensor networks already play in both simulating quantum
computation and synthesizing quantum circuits~\cite{bib:berezutskii2025tnqc}, we
use them to compress the $V$-step and represent the evolving wavefunction
$\psit$. The wavefunction is stored as a matrix
product state (MPS; the same object is called a tensor train in the numerical
analysis literature~\cite{bib:oseledets2011}, a convention we do not use here):
one rank-$\le D$ core per spatial axis (physical dimension $N$), with memory
$\mathcal{O}(dND^2)$ in place of the $N^d$ dense grid, where $d$ is the spatial
dimension. The velocity is the gradient of a learned scalar potential, so the
flow is parameterized by an energy function rather than by a velocity field,
which makes the continuity Hamiltonian well-defined. The potential step works
from the real \emph{generator} $V_t$ rather than any explicit operator: tensor
cross interpolation (TCI)~\cite{bib:ttcross, bib:xfac2025} interpolates $V_t$,
and a time-dependent variational principle (TDVP) step then applies
$e^{i\beta V_t}$ implicitly, growing the bond adaptively where the integrator
allows (Figure~\ref{fig:overview}c,d). Matrix product states suit this role for
two reasons. First, their time-evolution algorithms are established
machinery~\cite{bib:lubich2014, bib:haegeman}, giving an effective classical
ansatz for testing wavefunction flows beyond the dense-grid regime. Second,
under explicit gate, precision, and connectivity assumptions, a compact
final-state MPS is a candidate input for sequential state preparation: the state
is built on a number of ancilla qubits logarithmic in the bond dimension by a
depth-linear sequence~\cite{bib:schon2005seq, bib:ran2020encoding}, a route
already demonstrated on hardware for distribution-encoding
states~\cite{bib:iaconis2024}. Such an MPS also realizes exactly the
distribution class of polynomially deep tensor network
circuits~\cite{bib:huggins2019tn, bib:glasser2019tnfactor}, and a classically
optimized one can initialize a parametrized quantum circuit, mitigating
barren plateaus and warm-starting circuit
training~\cite{bib:rudolph2023synergistic}. The compressed object is thus a
classical precursor to a quantum circuit, not only a target for exact loading.

We make three contributions. First, to the best of
our knowledge, we present the first numerical study of wavefunction flows, which until now has been a purely theoretical proposal. Representing
the evolving wavefunction as a tensor network compresses the one
step whose cost otherwise grows exponentially with dimension: on
Gaussian mixture targets from two to eight dimensions, our method
matches exact (dense-grid) simulation in sample quality while using
$\sim\!10^7\times$ less memory and $\gtrsim\!10^3\times$ less evolution
wall-clock time, and runs in dimensions the dense-grid method cannot reach. Second, we
demonstrate tensor network evolution with a trained MPS
potential: rather than rebuilding the phase operator at every step, we
train the velocity potential once, directly as a real-valued MPS, and
consume its cores in the TDVP evolution. This gives an end-to-end \emph{classical}
tensor network pipeline and enables the orthogonal Gaussian mixture study up to
$d=32$. Third, we validate our pipeline by reproducing the
$\mathcal{O}(1/\sqrt{\prare})$ scaling of rare-event sampling. The rare-event
mass survives compression---the prepared state populates the far tail of every
mode of the target---and drawing from those tails by amplitude amplification
costs $2.5\times$ fewer state preparations per accepted rare sample than
classical rejection sampling.

Together, these results establish tensor networks as a practical numerical
framework for wavefunction flows and expose the low-rank regime in which the
method is classically simulable. The compact states produced by the
calculation may also serve as inputs to future state preparation studies.

\section{Results}\label{sec:results}

\emph{Datasets.} We study three target distributions of increasing
geometric complexity. The two-dimensional Swiss roll is a curved,
multi-scale manifold that stresses transport accuracy along a thin,
winding support. The Laplace distribution carries an amplitude cusp at
the origin whose non-smoothness sets a nontrivial $N^{-3/2}$ spectral
error floor and fixes the grid resolution requirement. The orthogonal
Gaussian mixtures place $2d$ well-separated modes on the coordinate
axes and extend systematically with dimension, providing the setting for the
dimensional scaling study ($d{=}2$--$8$, and to $d{=}32$ for the flow under the trained MPS potential) and for the downstream rare-event sampling demonstration.

\emph{Models.} Each target is evolved as a wavefunction flow under the continuity Hamiltonian
$\Hcons = i[\Kop,\Vt]$, Trotterized into
alternating kinetic ($\Kop$) and potential ($V$) steps. The kinetic step
is diagonal in the Fourier basis, so the classical cost concentrates in
the position-diagonal $V$-step, which we compress with tensor networks.
A single run fixes the per-axis grid size $N$, the Trotter step count
$K$, and the bond cap $D_{\max}$, evolves the state, and draws samples
directly from the final Born density $|\Psi_T|^2$.
We write $D$ for an imposed bond-dimension cap and $\chi$ for the bond
dimension a run actually attains, so that $\chi \le D_{\max}$ by
construction.
We benchmark the TCI+1TDVP and TCI+2TDVP integrators, and a
pre-trained MPS potential fed directly to the V-step with no per-step TCI
rebuild of $\Vt$, against a dense-grid
pseudospectral baseline and against the classical joint action matching (JAM)~\cite{bib:jam2023} sampler the flow emulates. We score sample
quality using the sliced Wasserstein distance (SW) and maximum mean
discrepancy (MMD); the prepared state then feeds into amplitude amplification for the
rare-event task. The results proceed from
dense-grid validation in two dimensions (Section~\ref{sec:results-2d}),
through sample quality and cost scaling across dimension
(Sections~\ref{sec:results-tsne} and~\ref{sec:results-cost}), to the
fully trained tensor network pipeline up to $d{=}32$
(Section~\ref{sec:results-bypass}) and its rare-event readout
(Section~\ref{sec:results-rare}).

\begin{figure}[tbp]
\centering
\includegraphics[width=0.95\textwidth]{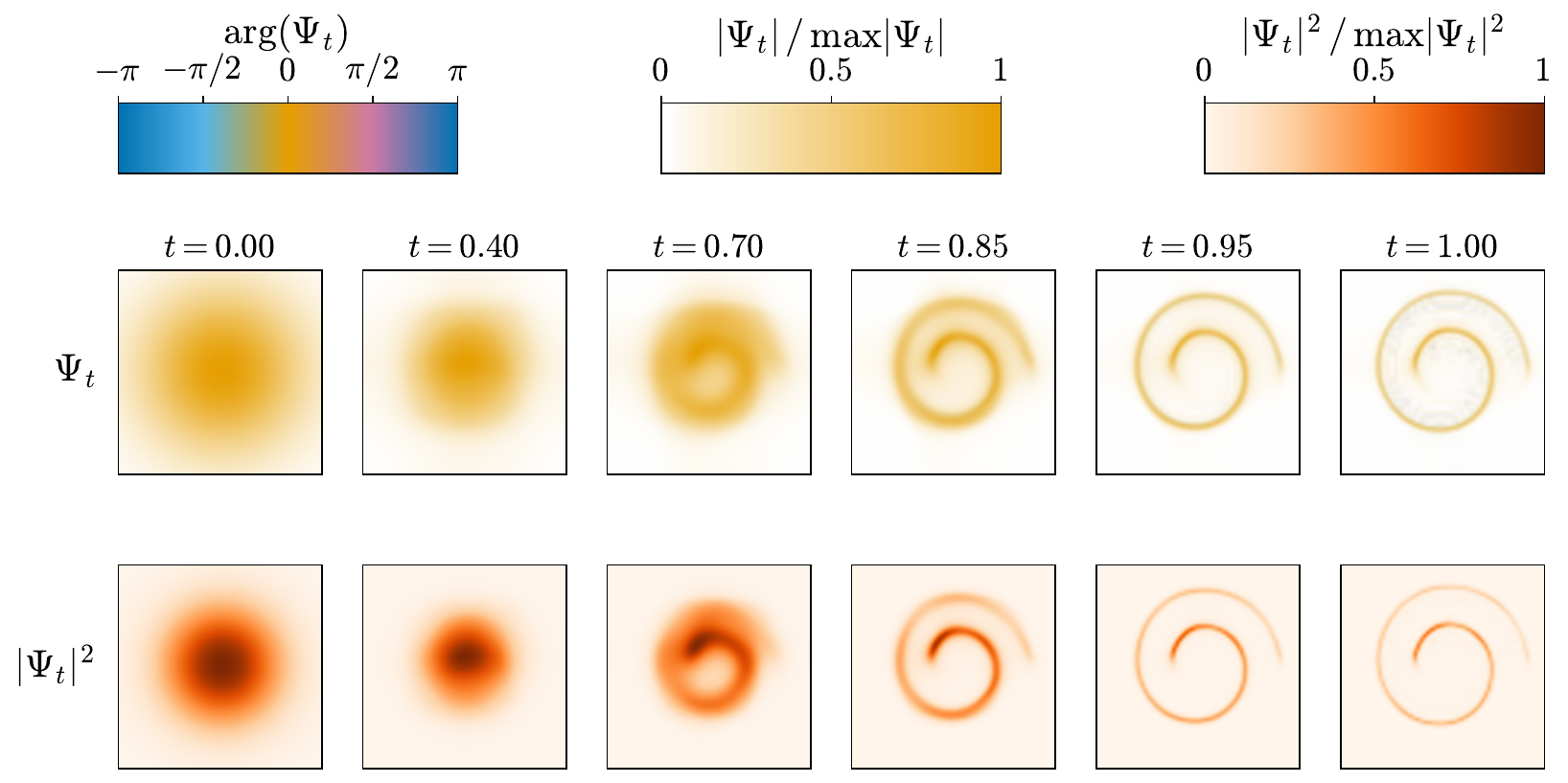}
\caption{\textbf{Dense-grid wavefunction evolution at $N{=}64$ across
generation time for the two-dimensional Swiss roll target.} Top row: the
wavefunction $\Psi_t$ visualized by hue $\!=\!\arg\Psi_t$ and
brightness $\!=\!|\Psi_t|/\max|\Psi_t|$. Bottom row: the probability
mass $|\Psi_t|^{2}$. Columns show
$t \in \{0,\,0.40,\,0.70,\,0.85,\,0.95,\,1.00\}$, spaced more finely
toward $t{=}1$ where the transport concentrates. The top-row hue
($\arg\Psi_t$) stays near $0$ throughout, consistent with the trivial-phase evolution of the wavefunction flow (Supplementary Section~\ref{app:a1}). The dynamics
under $\Hcons = i[\Kop, V_t]$ smoothly
transport the Gaussian source into the prescribed target
distribution. SW and MMD reported in the remaining figures of
the paper are computed against samples drawn from
$|\Psi(t{=}1)|^{2}$.}
\label{fig:viz2d}
\end{figure}

\subsection{Dense-grid simulation and error scaling}
\label{sec:results-2d}
We first validate the wavefunction flow pipeline on two
interpretable two-dimensional targets using direct (dense-grid wavefunction)
simulation: we evolve the full wavefunction on the $N^d$ grid with no
tensor network truncation, so that only the physical grid and Trotter
errors are present. On the Swiss roll
manifold (Figure~\ref{fig:viz2d}), the dense-grid state $\Psi_t$ evolves
smoothly under $\Hcons = i[\Kop, V_t]$ and the final mass
$|\Psi_T|^2$ reproduces the target. For the two-dimensional orthogonal Gaussian
mixture, where the velocity $v_t = \nabla V_t$ is available in
closed form, the equivalent classical probability flow picture (Figure~\ref{fig:gmm-ode}) integrates source samples under
$\dot{x}_t = v_t(x_t)$ onto the target's four modes along the radial
transport encoded by the analytic $V_t$; this is the continuum limit
of the flow, $|\Psi_T|^2 = p_1$ as an exact identity
(Supplementary Section~\ref{app:a1}). Sample quality is measured using the sliced Wasserstein distance and
maximum mean discrepancy~\cite{bib:gretton2012}. Both are distances between two sample sets (Eq.~\eqref{eq:sw}); throughout we
abbreviate $\mathrm{SW}(q) \equiv \mathrm{SW}(\{x\}_q, \{x\}_{p_1})$, where $q$ is
any time-dependent source and $\{x\}_q$ its samples---Born samples when $q$ is a
wavefunction, transported samples when $q$ is a marginal density. A single
argument therefore always names the source being scored against the target.

The dense-grid flow also lets us compare the grid-free error terms against
their predicted rates before compression is applied
(Figure~\ref{fig:scaling-bounds}). On exact $d{=}2$ evolutions, the
spatial/grid error $\varepsilon_{\mathrm{grid}}$ follows the Sobolev
rate $N^{-3/2}$ fixed by the kink in the Laplace target amplitude
($\Psi_T \in H^{s}$ for every $s<3/2$), while the Trotter error
$\varepsilon_{\mathrm{time}}$ follows the global $O(K^{-1})$
order---the accumulation of the local $O(\Delta t^2)$ product-formula
step---with a dimensional constant that stays within the $d^2$ bound;
both are consistent with the wavefunction flow error
theory~\cite{bib:layden2025}. These are two of the three additive
error sources that organize our study---grid, Trotter, and
compression $\sum_k \eta_k$---the third of which enters only with the
matrix product state compression. In the dimension sweeps of
Sections~\ref{sec:results-tsne}--\ref{sec:results-cost}, we report
the total error via SW, and we isolate the compression term in
Section~\ref{sec:results-bypass}.

\begin{figure}[tbp]
\centering
\includegraphics[width=0.95\textwidth]{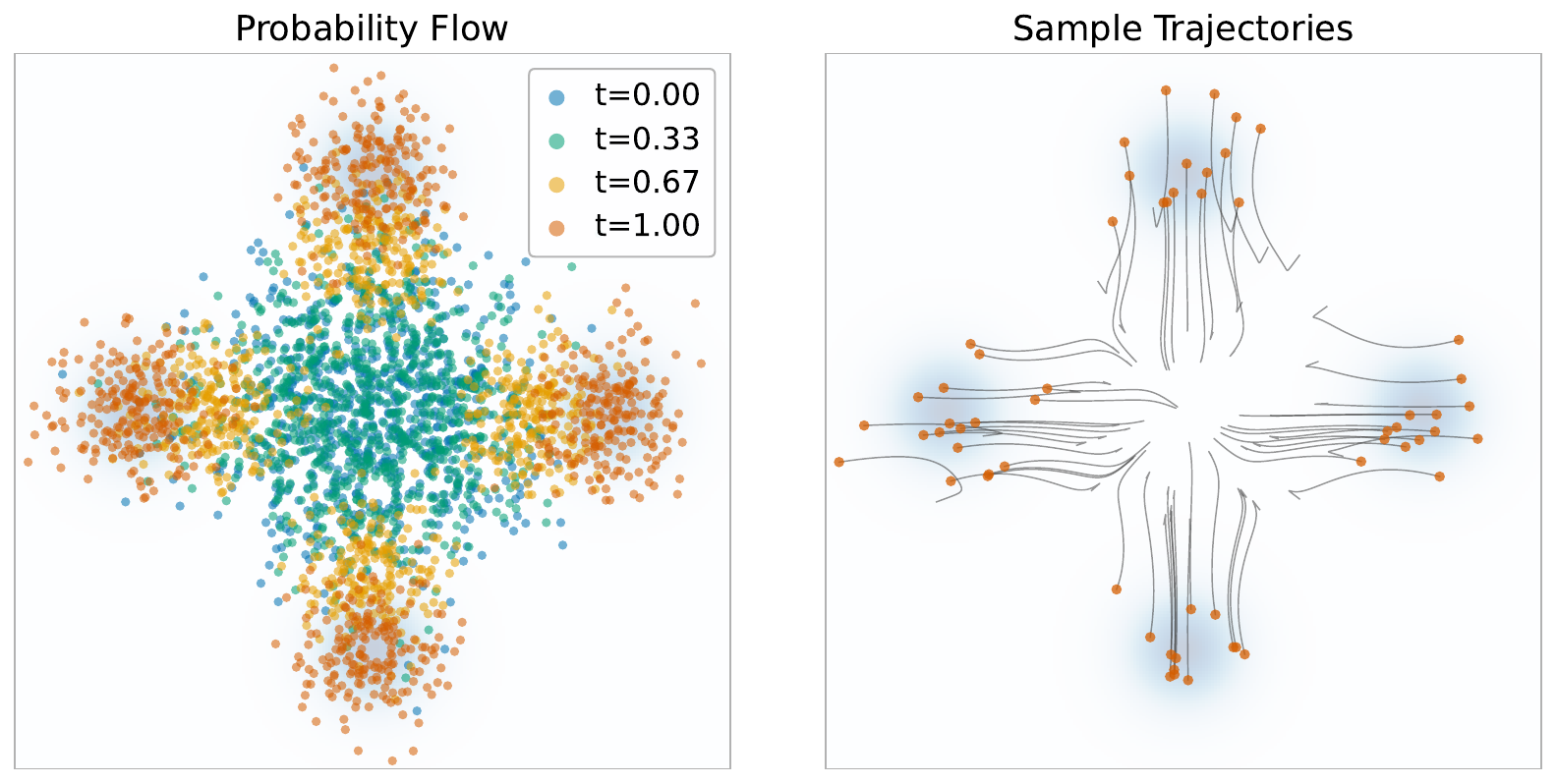}
\caption{\textbf{Probability flow view of the two-dimensional orthogonal Gaussian mixture target.} Source samples $x_0 \sim \mathcal{N}(0,
\sigma_0^2 I)$ are integrated under $\dot{x}_t = v_t(x_t)$, with
$v_t = \nabla V_t$ taken from the analytic potential of
Supplementary Eq.~\eqref{eq:Vt-gmm}. (Left) Sample positions at four time
points $t \in \{0, 0.33, 0.67, 1.0\}$, colored by time and
overlaid on the faded target density. (Right) Trajectories of
the same samples from $t=0$ to $t=1$ as gray lines, with
$t=1$ endpoints highlighted; the radial structure of the
transport is set by the four orthogonal modes of the target.
This is the classical limit of the wavefunction flow:
$|\Psi_T|^2 = p_1$ as an exact identity in the continuum (see
Supplementary Section~\ref{app:a1}).}
\label{fig:gmm-ode}
\end{figure}

\begin{figure}[tbp]
\centering
\includegraphics[width=\textwidth]{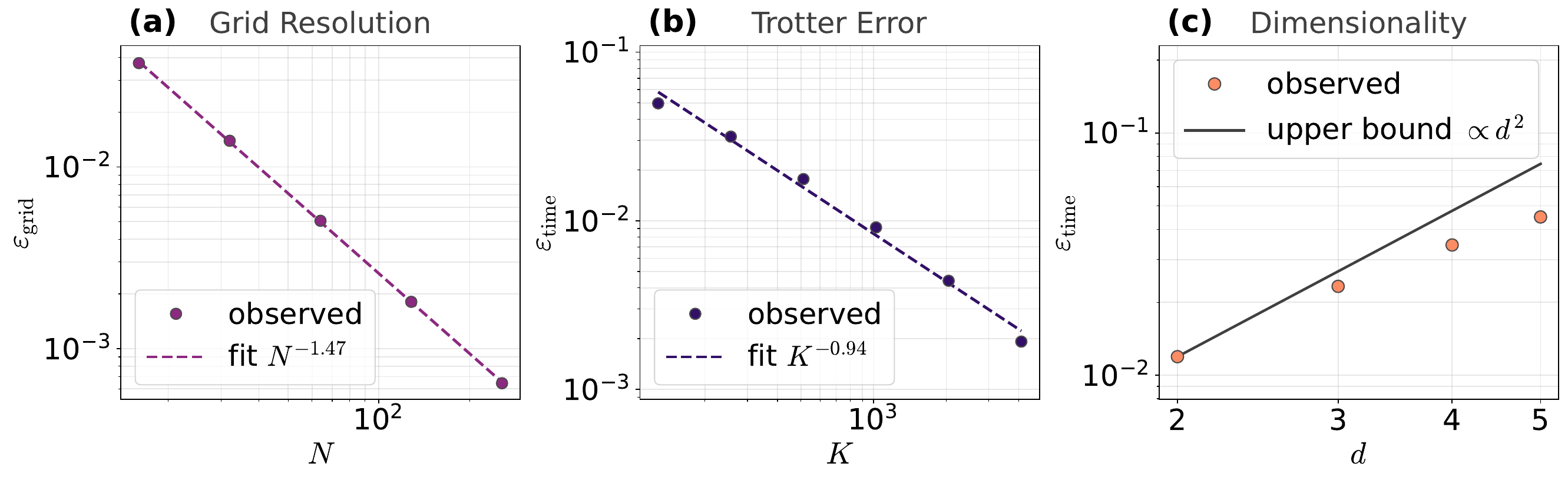}
\caption{\textbf{Error scaling of the dense-grid wavefunction flow matches the
predicted decomposition.} Two grid-free error sources, isolated on exact
$d{=}2$ evolutions (no tensor network truncation) and measured with a
phase-aligned state metric. (a)~Grid/representation error
$\varepsilon_{\mathrm{grid}} = \|(1{-}P_N)\Psi_T\| / \|\Psi_T\|$, the relative
$L^2$ Fourier-truncation error of the ideal target amplitude
$\Psi_T \propto \sqrt{p_T}$, versus grid resolution $N$. For a Laplace target
the amplitude $\sqrt{p} = e^{-|x|}$ has a kink, so $\Psi_T \in H^{s}$ for
every $s<3/2$ and the derivative discontinuity produces Fourier coefficients
consistent with $\varepsilon_{\mathrm{grid}} \sim N^{-3/2}$ (fit $N^{-1.47}$). (b)~Trotter
error $\varepsilon_{\mathrm{time}} = \|\Psi_K - \Psi_\infty\|$, the state distance
between the $K$-step product-formula flow and its $K\to\infty$ limit, versus
Trotter steps $K$ (at $N{=}16$, where $K \gg N^2$ is reachable). The global order
is $O(K^{-1})$ (fit $K^{-0.94}$)---the accumulation of the local $O(\Delta t^2)$
step error over $K = T/\Delta t$ steps. (c)~The same
$\varepsilon_{\mathrm{time}}$ versus dimensionality $d$ (fixed $K,N$): the realized
growth ($\sim d^{1.5}$) stays within the $d^2$ coefficient upper bound. Dense-grid
propagation isolates these two terms; the tensor network ansatz adds a third,
compression $\sum_k \eta_k$ (Section~\ref{sec:results-bypass}).}
\label{fig:scaling-bounds}
\end{figure}

\subsection{Sample quality across spatial dimensions}
\label{sec:results-tsne}
In Figure~\ref{fig:tsne} we overlay the generated TCI+1TDVP samples on
the true target---together with an independent target draw (labeled
\emph{Exact}) as a finite-sample reference---at $d \in \{3, 5, 7\}$, in a panel-local
t-SNE~\cite{bib:tsne2008} embedding. At every $d$, the TDVP samples
reproduce the target mode structure, resolving the same number of
distinct clusters as the target and achieving visually comparable occupancy. We
read t-SNE only qualitatively: it distorts global inter-cluster
distances and densities, so neither cross-panel coordinates nor
inter-cluster ratios are quantitative (quantitative accuracy is the
per-panel sliced Wasserstein, annotated for each panel). The annotation is
the plotted run's own SW; in Table~\ref{tab:sw} we report the best cell
by seed-averaged SW, which need not be the same run.

\begin{figure}[tbp]
\centering
\includegraphics[width=0.95\textwidth]{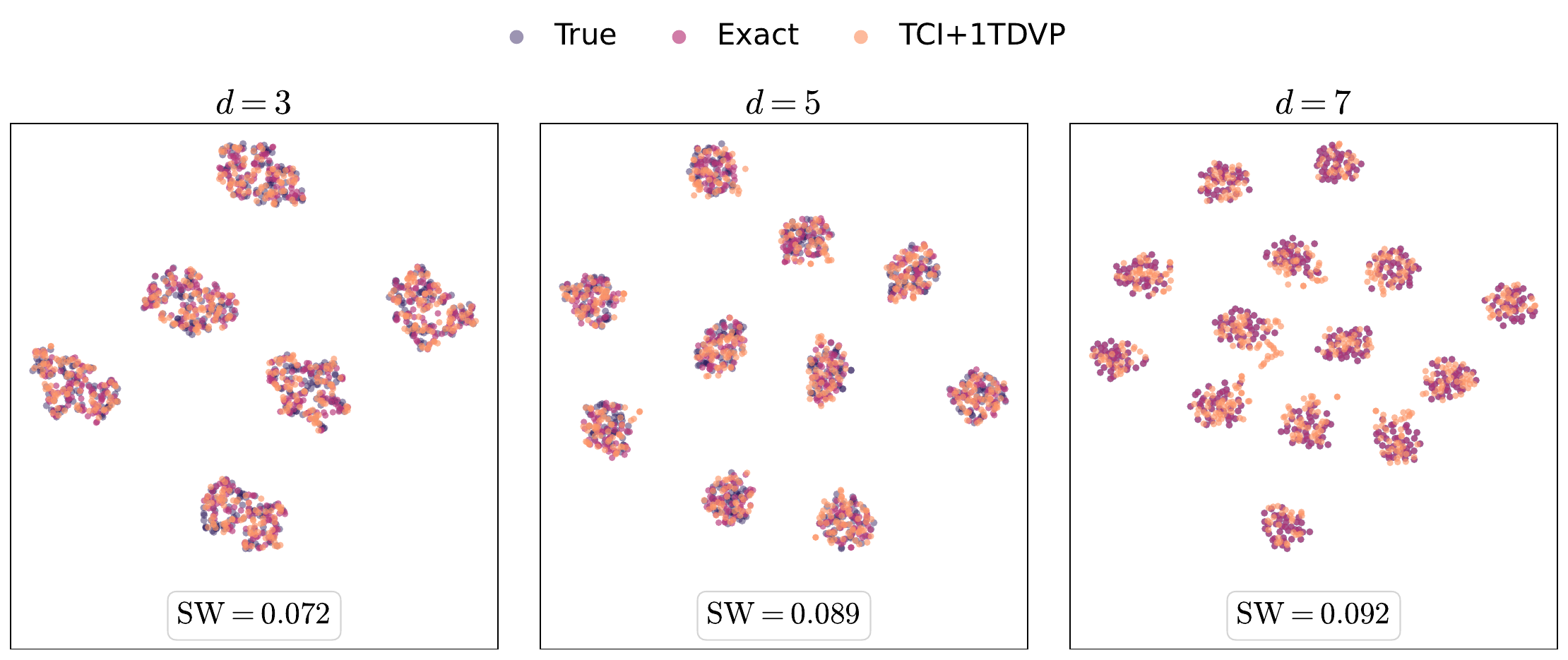}
\caption{\textbf{Low-dimensional visualization (t-SNE) overlays at $d = 3, 5, 7$ for orthogonal
Gaussian mixtures.} Each panel fits a panel-local t-SNE and overlays,
on a magma scale from darkest to lightest, the true target samples
(darkest), an independent target draw as a finite-sample reference
(labeled \emph{Exact}), and the generated TCI+1TDVP samples (lightest); the
sliced Wasserstein of the plotted run is annotated. That is a per-run value
and need not coincide with Table~\ref{tab:sw}, which minimizes over cells
after averaging across seeds. The TCI+1TDVP samples resolve the same
cluster structure as the target at every $d$; the embedding is
qualitative (t-SNE does not preserve global distances or density), and
the quantitative agreement is the annotated sliced Wasserstein.}
\label{fig:tsne}
\end{figure}

Quantitatively, we report in Table~\ref{tab:sw} the best sliced Wasserstein
across the hyperparameter sweep for the dense-grid baseline and the TDVP
V-step variants. Because each cell is estimated from at most three seeds,
minima taken over the search box are biased low by cell selection; the
replicated $n{=}10$ audit of Supplementary Section~\ref{app:audit} addresses this for
JAM and TCI+1TDVP by re-evaluating located optima over ten fresh seeds
(under a trained rather than analytic potential) and confirms the
ordering. TCI+1TDVP is at least as accurate as the JAM classical
baseline at six of the seven dimensions, and its worst case is a
factor of $1.06$ (at $d{=}8$).
The TCI+1TDVP best-cell SW values stay within a factor of $1.7$
of the target--target reference across the full
$d=2{-}8$ range ($0.115, 0.081, 0.095, 0.092, 0.093, 0.092, 0.100$),
with no upward drift in $d$. We caveat that orthogonal
Gaussian mixture targets may not fully stress-test the method at
high $d$: the support of the target mixture occupies an exponentially
small fraction of the $[0,L)^d$ volume, so the relative target
structure becomes simpler as $d$ grows. The trained MPS potential scaling study (Section~\ref{sec:results-bypass}) pushes this same orthogonal
target to $d = 32$ under a learned MPS velocity potential; a more
decisive stress test on a fully high-dimensional target with rich
correlations is left to future work.

\begin{table}[!htbp]
\centering
\caption{\textbf{Best-cell sliced Wasserstein for the dense-grid baseline
and TDVP V-step variants, by spatial dimension.} Values are the minimum
over the hyperparameter
search box (mean over $\le 3$ seeds when available); -- means
the method is infeasible or unsampled at that $d$. The
target--target reference is the SW between two independent
finite samples of the target ($n=500$ samples, 128 random
projections; see Section~\ref{app:b8}) --- a finite-sample scale, not
an irreducible floor; the method rows use the same $n=500$. Cells cluster
tightly near the minimum, and the gap
between the two tensor network integrators is comparable to the spread of an
individual cell (median $0.005$ across cells, combining seed-to-seed
variation with the Monte Carlo spread of the SW estimator itself), so over
the searched box TCI+1TDVP and TCI+2TDVP are not distinguishable on
accuracy, consistent with Section~\ref{sec:discussion}, which separates
them on cost. A separate
replicated $n{=}10$ audit under a trained velocity potential
(Supplementary Section~\ref{app:audit}, JAM and TCI+1TDVP only) validates, with a
second metric (MMD), that TCI+1TDVP tracks the JAM baseline. That audit
includes neither Dense-grid nor TCI+2TDVP, so it does not re-rank them.}
\label{tab:sw}
\begin{tabular}{lccccccc}
\toprule
Method & $d=2$ & $d=3$ & $d=4$ & $d=5$ & $d=6$ & $d=7$ & $d=8$ \\
\midrule
Target--target (finite-sample) & 0.069 & 0.068 & 0.066 & 0.065 & 0.063 & 0.062 & 0.060 \\
JAM & 0.131 & 0.143 & 0.133 & 0.102 & 0.106 & 0.101 & 0.094 \\
\midrule
Dense & 0.107 & 0.092 & 0.086 & 0.104 & -- & -- & -- \\
TCI+1TDVP & 0.115 & 0.081 & 0.095 & 0.092 & 0.093 & 0.092 & 0.100 \\
TCI+2TDVP & 0.118 & 0.084 & 0.098 & 0.099 & 0.111 & 0.091 & 0.100 \\
\bottomrule
\end{tabular}
\end{table}

\subsection{Scalable classical simulation via tensor networks}
\label{sec:results-cost}
Here, the tensor network $V$-step makes wavefunction flow simulation
scalable in $d$. In Figure~\ref{fig:cost} we trace the
best-cell trajectory of the TDVP V-step variants across
$d \in \{2, \ldots, 8\}$ relative to the dense-grid baseline. Two ratios
measure the saving. The \emph{compression ratio}
$\mathcal{R} = N^d/(2N\chi + (d{-}2)N\chi^2)$
(Figure~\ref{fig:cost}a) divides the dense-grid amplitude count by the MPS
parameter count at the achieved bond dimension $\chi$, and the
\emph{acceleration factor}
$\mathcal{S} = t_{\mathrm{Dense\text{-}grid}}/t_{\mathrm{MPS}}$
(Figure~\ref{fig:cost}b) divides the dense-grid evolution wall-clock time by
the tensor network's; values above unity therefore mean the tensor network is
the more compact and the faster representation, respectively. The
corresponding best-cell accuracy (SW) is shown in
Figure~\ref{fig:cost}c. The dense-grid baseline
becomes intractable above $d \!\approx\! 5$: at $d{=}5$ a single
Dense-grid run already costs ${\sim}10^{4}$\,s and ${\sim}1$\,GB of
grid storage, and at $d \ge 6$ the $32^d$ amplitude tensor exceeds
container memory at the modest grid resolution we use. Within the
regime where the comparison is measured, TCI+1TDVP crosses the
cost--accuracy Pareto frontier of Dense-grid near $d{=}4$ and
TCI+2TDVP near $d{=}5$; above these crossovers, the tensor network
methods are cheaper than Dense-grid at matching SW, the one exception being
TCI+2TDVP at $d{=}6$ ($0.9\times$), where the Dense-grid wall-clock time is
projected rather than measured. The MPS
compression factor rises by roughly an order of magnitude per added
dimension. At $d{=}8$, the TCI+1TDVP best cell uses approximately
seven orders of magnitude less memory and at least three orders of
magnitude less evolution wall-clock time than an extrapolated dense-grid baseline
(cost models in Supplementary Section~\ref{app:a6}); the
$d{=}8$ dense-grid numbers are extrapolated from the measured $d \le 5$ scaling,
because Dense-grid is itself infeasible there. The underlying cost--accuracy
Pareto frontiers (at $d{=}5$ and faceted across $d$) are shown in
Supplementary Section~\ref{app:supp-results} (Figure~\ref{fig:supp-pareto}).
In this analytic-$V$ sweep (runtime-TCI V-step) the Pareto-optimal cells sit at the
search-grid maxima ($K=64$ Trotter steps, $D=32$ at $d=8$): within the
searched box, the residual Trotter error $\mathcal{O}(1/K)$ still dominates
the slowly accumulating compression error $\sum_k\eta_k$, so accuracy
improves monotonically to the largest $K$ tried and a larger box may yield
further gains. This is a different regime from the trained MPS potential audit
of Supplementary Section~\ref{app:audit}, where coordinate descent over a wider $K$ range
locates an interior optimum $K^\star$ once accumulated compression
error overtakes the Trotter gain; both observations reflect the same tradeoff
between Trotter error and accumulated compression error.

\begin{figure}[tbp]
\centering
\includegraphics[width=0.98\textwidth]{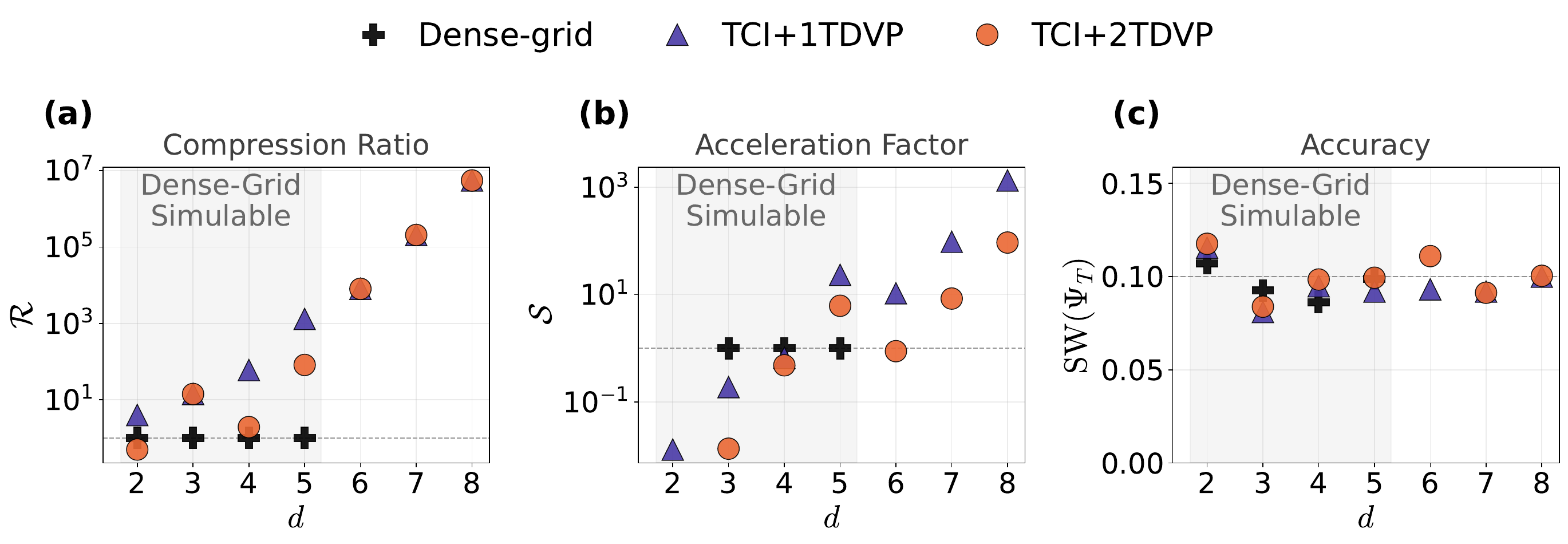}
\caption{\textbf{Tensor network costs scale polynomially in $d$ at
fixed accuracy.} Best-cell trajectory of the TDVP V-step variants
(minimum mean sliced Wasserstein over the hyperparameter search box,
analytic $V_t$) relative to the Dense-grid baseline across
$d \in \{2,\ldots,8\}$; the underlying data are the
same runs that produce Table~\ref{tab:sw}. (a)~Compression ratio
$\mathcal{R} = N^d/(2N\chi + (d{-}2)N\chi^2)$, the Dense-grid amplitude count
over the MPS parameter count at the achieved bond dimension $\chi$;
$\mathcal{R}$ rises by roughly an order of magnitude per
added $d$, reaching ${\sim}10^{7}$ at $d{=}8$. The shaded band in
all three
panels marks the regime where the Dense-grid baseline is feasible and
was run ($d \le 5$); every tensor network marker is a measured run.
(b)~Evolution acceleration factor
$\mathcal{S} = t_{\mathrm{Dense\text{-}grid}}/t_{\mathrm{MPS}}$, the Dense-grid evolution
wall-clock time over the tensor network's, so that values
above unity mean the tensor network is faster. The wall-clock time is summed per Trotter step and excludes the per-step metric
callback; that callback is under $3\%$ of the enclosing loop at $d\ge4$
but over $90\%$ at $d{=}2$, enough to invert the apparent ordering of
$d{=}2$ and $d{=}3$ if included. TCI+1TDVP crosses above unity near
$d{=}4$ and TCI+2TDVP near $d{=}5$, after which the
tensor network methods are cheaper than Dense-grid at matching
accuracy, the one exception being TCI+2TDVP at $d{=}6$ ($0.9\times$).
The Dense-grid numerator is measured at $d\in\{3,4,5\}$ and is a
log-linear fit through those three points elsewhere, the $N^d$ grid being
infeasible above $d{=}5$; the $d{=}6$ exception falls in the projected
regime. The tensor network denominators are
measured from runs at each shown $d$.
(c)~Best-cell accuracy $\mathrm{SW}(\Psi_T)$ (sliced Wasserstein of the
endpoint state) versus $d$: TCI+1TDVP and
TCI+2TDVP stay within a factor of $1.7$ of the target--target
reference across the full range, without upward drift in $d$, and track
the Dense-grid baseline wherever it is feasible. The cost--accuracy Pareto frontiers are given in the
supplement (Supplementary Section~\ref{app:supp-results},
Figure~\ref{fig:supp-pareto}); a replicated $n{=}10$
audit with a second metric appears in Supplementary Section~\ref{app:audit}.}
\label{fig:cost}
\end{figure}

\subsection{Tensor network evolution with a trained MPS potential}
\label{sec:results-bypass}
The most direct realization of a wavefunction flow pre-trains the velocity potential as a tensor network, removing the runtime TCI construction entirely;
the choice of integrator then sets the balance between cost and accuracy. In
Sections~\ref{sec:results-2d}--\ref{sec:results-cost}, the V-step applies $\exp(\mathrm{i}\beta V_t)$ by building a fresh TCI approximation of the \emph{generator} $V_t$ at every Trotter substep (required because $V_t$
is time-dependent) and letting TDVP apply its exponential implicitly ---
the exponential itself is never crossed, precisely because the rank of
the unitary can inflate far above that of the generator
(Supplementary Remark~\ref{rem:tdvp-vs-tci}). That per-substep TCI build of the generator is the
dominant per-step cost. We remove it
entirely by training the velocity potential once, directly as a
real-valued MPS, and handing its cores to the V-step. Because
the learned cores are real, they define a real diagonal generator whose exact
exponential $\exp(\mathrm{i}\beta V_t)$ is unitary; the TDVP V-step approximates
its action on the MPS manifold directly from the $V_t$ cores contracted against
the wavefunction boundary tensors, with no runtime TCI at all (the trained MPS $V$-step; Section~\ref{app:trained-v}). We parameterize the trained oracle as
a \emph{trained MPS potential} (one tensor site per spatial axis, local dimension
$N$), fit via joint action matching (JAM; Supplementary Section~\ref{app:a3}), using the
same objective as the classical baseline.

\begin{figure}[t!]
\centering
\includegraphics[width=0.95\textwidth]{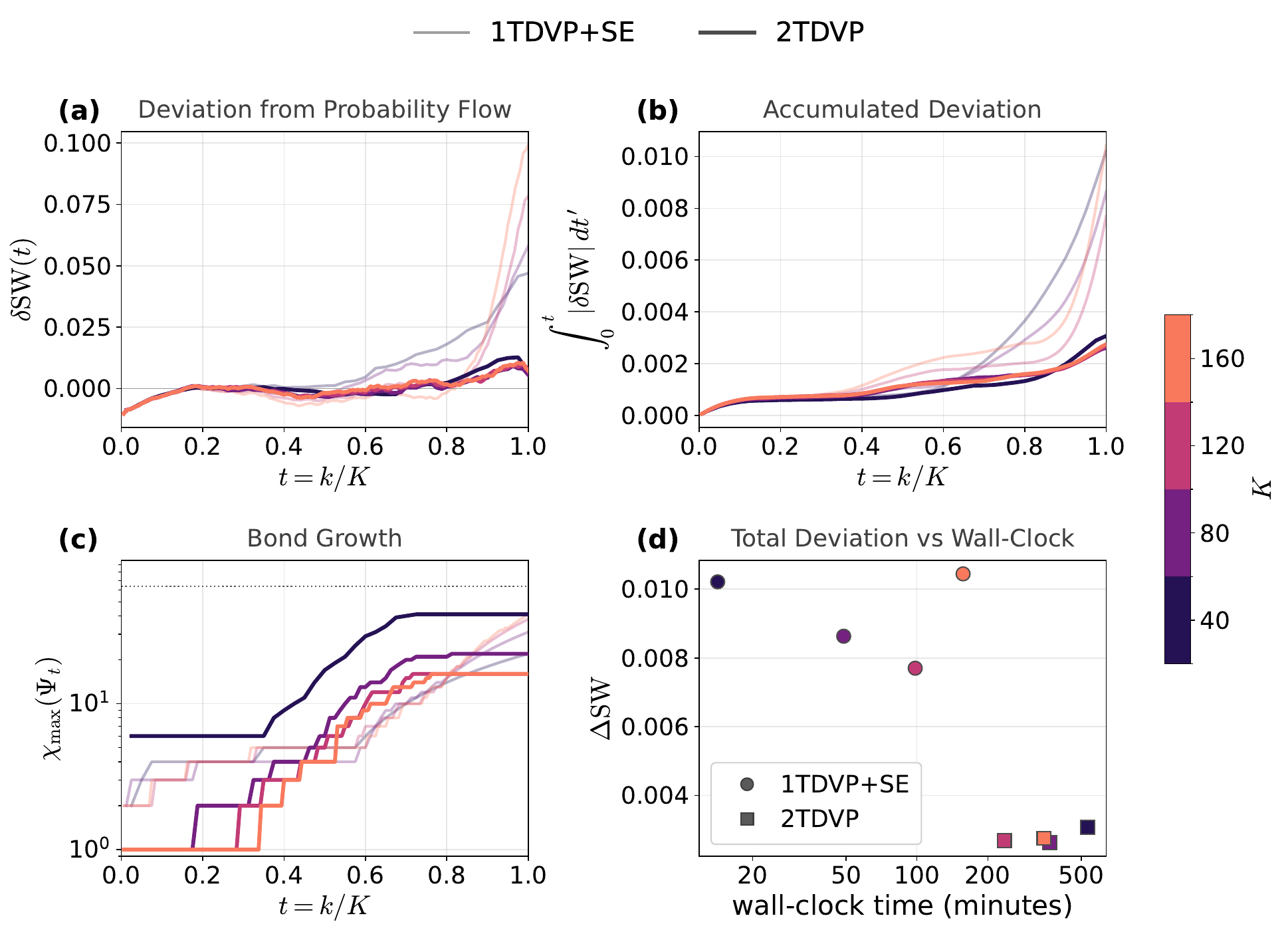}
\caption{\textbf{Deviation from the probability flow grows as the flow
reaches the target and the V-step integrator sets the cost--accuracy
balance.} Gaussian mixture at $d{=}8$, $\sigma{=}0.5$; the trained MPS potential
cores are fed directly to the V-step, with no runtime TCI.
Two encodings
run through panels (a--c): color denotes the Trotter count
$K \in \{40, 80, 120, 160\}$ (dark to light), and line style denotes
the integrator---pure 2TDVP (two-site update every Trotter step, thick lines)
versus 1TDVP with demand-driven subspace expansion (1TDVP+SE, thin lines);
panel~(d) instead marks the integrator by symbol (circles: 1TDVP+SE;
squares: 2TDVP). (a)~Instantaneous deviation from the probability flow,
$\delta\mathrm{SW}(t) = \mathrm{SW}^{\mathrm{wf}}(\psit) - \mathrm{SW}^{\mathrm{pf}}(p_t)$,
where $\mathrm{SW}(q)$ is the distance from samples of the source $q$ to samples
of the target---$\psit$ sampled by Born measurement, the marginal density $p_t$
by transport---the two track closely until late in the flow;
(b)~its cumulative integral $\int_0^t |\delta\mathrm{SW}|\,dt'$.
(c)~Bond dimension $\chi_{\max}(\Psi_t)$---pure 2TDVP
self-limits $\chi$ by truncation ($\chi^\star\!\approx\!16$) while
1TDVP+SE grows $\chi$ at the larger $K$ needed to suppress Trotter error.
(d)~Total deviation $\Delta\mathrm{SW} = \int_0^1 |\delta\mathrm{SW}|\,dt$ versus
measured wall-clock time: a cost--accuracy frontier---pure 2TDVP attains lower deviation at a modest ($\sim\!2$--$7\times$)
wall-clock cost over 1TDVP+SE (up to $\sim\!10$--$40\times$ at small $K$).}
\label{fig:gmm-bypass}
\end{figure}

In Figure~\ref{fig:gmm-bypass} we compare the two V-step integrators under the
trained cores on the $d{=}8$ Gaussian mixture. Write
$\delta\mathrm{SW}(t) = \mathrm{SW}^{\mathrm{wf}}(\psit) - \mathrm{SW}^{\mathrm{pf}}(p_t)$
for the instantaneous deviation of the wavefunction flow from the probability
flow, where $p_t$ is the marginal density the probability flow transports, and
$\Delta\mathrm{SW} = \int_0^1 |\delta\mathrm{SW}(t)|\,dt$ for the deviation
accumulated over the whole flow. Pure 2TDVP truncates the
subspace-expansion roughness every Trotter step, so it self-limits its
bond (Figure~\ref{fig:gmm-bypass}c; $\chi^\star\!\approx\!16$) and reaches
a lower final sliced Wasserstein, whereas 1TDVP+SE must grow $\chi$ at the
larger $K$ needed to control the Trotter splitting
(Figure~\ref{fig:gmm-bypass}a,b,d; Section~\ref{app:se-adaptive}). The accuracy of the
$V$-step is thus limited by the Trotter-step count rather than by the
bond capacity.

\paragraph{Scaling the trained MPS potential flow}
Feeding the trained cores to the 2TDVP V-step (no runtime TCI), we evolve the full wavefunction flow on the orthogonal Gaussian mixture and measure how far this flow scales while staying faithful to the target. In
Table~\ref{tab:scaling} we report, for the trained MPS potential across
$d \in \{8, 12, 16, 32\}$ at fixed $K{=}160$, the endpoint
sliced Wasserstein SW$^{\mathrm{wf}}_{T}$ (mean $\pm$ 95\% confidence
interval over $10$ independent initializations of the trained MPS
potential), the probability flow
action matching floor SW$^{\mathrm{pf}}_{T}$, and the achieved bond
$\chi^\star$. Up to $d{=}32$, the flow stays within a small,
dimension-stable factor of that floor, with a self-limiting
bond that stays below the cap ($\chi^\star < D_{\max}{=}64$ across all
replicates) and does not grow with $d$, so the cap is never hand-tuned. This is an
orthogonal Gaussian mixture scaling result, not a claim about high-dimensional
correlated structure (the caveat of Section~\ref{sec:results-tsne}). The compact bond dimensions here enable efficient
simulation and representation within this family; they do not, by themselves,
establish a preparation circuit or a quantum speedup.

\begin{table}[tbp]
\centering
\caption{\textbf{Trained MPS potential tensor network scaling on the orthogonal
Gaussian mixture.} The learned real-valued MPS velocity potential (bond $D{=}8$)
fed directly to the 2TDVP V-step, eliminating the runtime TCI construction; fixed
$K{=}160$ Trotter steps across $d$. Each row averages $n{=}10$
independent initializations of the trained MPS
potential. Columns: spatial dimension $d$;
oracle parameter count; endpoint sliced Wasserstein
SW$^{\mathrm{wf}}_{T}$ (mean $\pm$ 95\% CI); the probability flow
action matching floor SW$^{\mathrm{pf}}_{T}$ ($\dot x = \nabla V$); and
the achieved bond $\chi^\star$.}
\label{tab:scaling}
\begin{tabular}{lrrrr}
\toprule
$d$ & params & SW$^{\mathrm{wf}}_{T}$ (mean $\pm$ 95\% CI) & SW$^{\mathrm{pf}}_{T}$ & $\chi^\star$ \\
\midrule
8 & 15k & 0.0363 $\pm$ 0.0027 & 0.0299 & 5--20 \\
12 & 23k & 0.0341 $\pm$ 0.0026 & 0.0293 & 6--31 \\
16 & 31k & 0.0380 $\pm$ 0.0020 & 0.0291 & 6--19 \\
32 & 64k & 0.0306 $\pm$ 0.0009 & 0.0250 & 7--11 \\
\bottomrule
\end{tabular}
\end{table}

\FloatBarrier

\subsection{Rare-event sampling and idealized amplification scaling}
\label{sec:results-rare}
Rare-event sampling is governed by the mass a distribution places many
standard deviations from its modes. We define a rare event by the standard
multivariate $k$-sigma rule: a point lies
in the tail of its nearest mode $c_j$ when $\lVert x - c_j \rVert > k\sigma$,
so that for well-separated orthogonal modes the tail probability is
$\prare(k) = \Pr[\chi^2_d > k^2]$, the upper tail of a $\chi^2_d$ variate (here
$d{=}8$, $k{=}4$, so $\prare \approx 4.24\%$). In Figure~\ref{fig:rare-event}(b,c) we embed Born samples from the
$K{=}160$ prepared state (Section~\ref{sec:results-bypass}) under a shared
t-SNE map.
The flow reproduces all $2d{=}16$ modes and populates the $>\!4\sigma$ shells around each mode, so rare-event mass survives the tensor network
evolution.

Drawing rare samples from this state is accelerated by amplitude
amplification~\cite{bib:bhmt}. Boosting the tail amplitude to
$\mathcal{O}(1)$ costs $\mathcal{O}(1/\sqrt{\prare})$ state preparations per
accepted rare sample, against the $\mathcal{O}(1/\prare)$ draws of classical
rejection sampling (Figure~\ref{fig:rare-event}a). Both pipelines are built
on the same learned potential---the flow ODE pipeline integrates $\nabla V$ of
the oracle that the wavefunction flow consumes---so the learning error is common to
both pipelines. The comparison isolates the transport and the sampling method.
Both are classical simulations; the difference lies in which transport prepares the
samples and how the tail is drawn from them. For the prepared
state, this costs $2.5\times$ fewer preparations per rare sample at $4\sigma$,
rising to $5.9\times$ at $5\sigma$ and $7.8\times$ at $6\sigma$; at a budget
of $500$ preparations, the difference is visible as $42$ against $104$ ringed
tail points in Figure~\ref{fig:rare-event}(b,c).

This is an idealized demonstration: we model the exact Grover
outcome distribution $P(n) = \sin^2((2n{+}1)\theta)$ after $n$ amplification
rounds, with $\prare = \sin^2\theta$
(Methods), evaluated at the rare-event amplitude the prepared state carries,
rather than compiling and executing a circuit-level $U_{\mathrm{prep}}$. The
marking oracle that flags the tail, its (polynomial) cost, and the scope of
the advantage---in particular that a low-rank observable would be classically
contractible---appear in Supplementary Section~\ref{app:rare-oracle}. This
demonstration establishes that the tensor network flow delivers a state on
which amplitude amplification reproduces its idealized cost scaling. It does
not establish a gate-level or
dequantization-resistant advantage on this benchmark, which would require a
controlled preparation circuit and an observable that is not efficiently
contractible.

\begin{figure}[tbp]
\centering
\includegraphics[width=\textwidth]{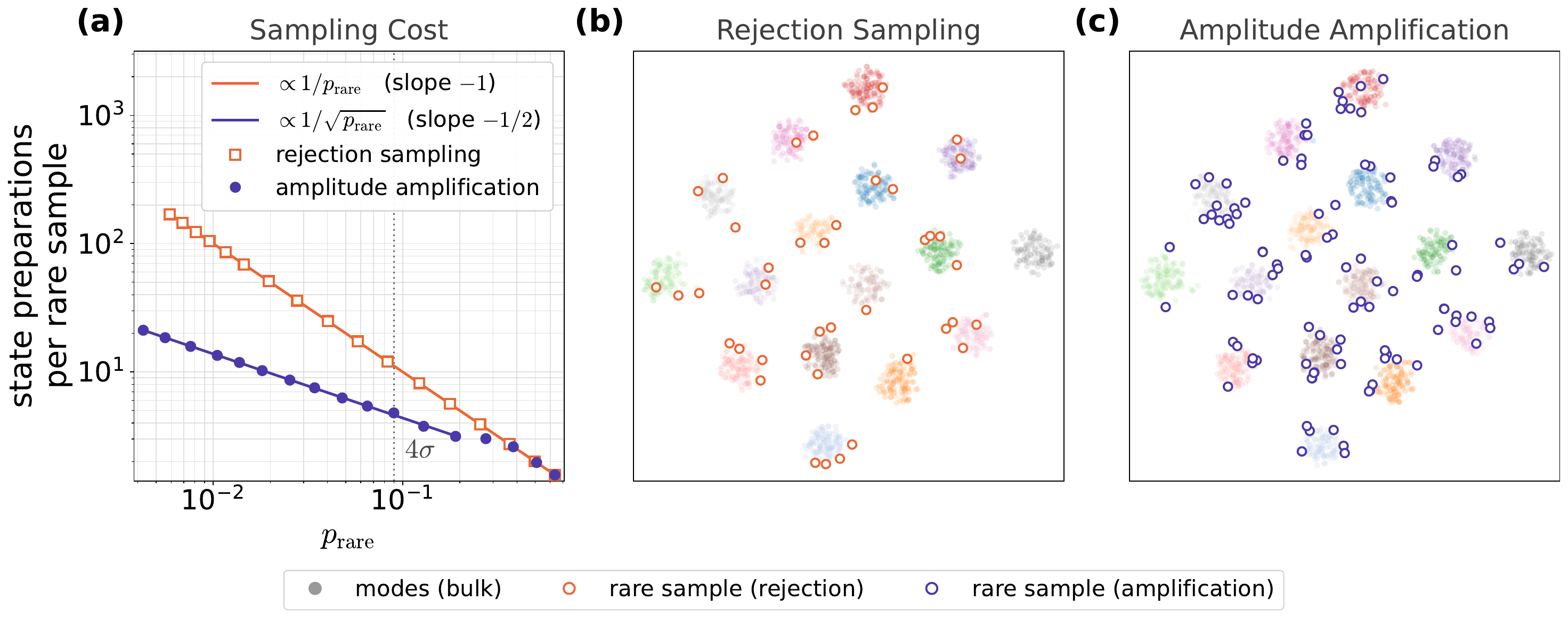}
\caption{\textbf{Rare-event sampling cost and yield on the $d{=}8$ orthogonal
Gaussian mixture ($\sigma{=}0.5$, $>\!4\sigma$, $a\approx4.24\%$).} Both pipelines are driven
by the same learned velocity potential---the flow ODE pipeline
integrates $\dot x = \nabla V$ of the oracle whose cores the wavefunction flow
consumes---so state preparation error is common to both pipelines. The comparison
isolates the transport and the sampling method. Both are classical
simulations. Cost is quoted in state
preparations, i.e.\ complete transports.
(a)~\emph{Sampling cost.} State preparations per accepted rare sample against
the rare-event probability $\prare$, each pipeline plotted at its own measured
tail mass. Markers are measurements and lines are the cost models. Prefactors
are fitted with the exponents held at their
structural values: rejection over all points, amplification on the tail
($\prare\le0.2$, i.e.\ $k=3.5$--$6.5$) where its $1/\sqrt{\prare}$ law is asymptotic;
the fitting protocol and the round-count rule are given in Methods.
(b)~\emph{Rejection sampling.} Drawn from the flow ODE state: at a
budget of $Q{=}500$ state preparations and a $4\sigma$ threshold, it yields
$Qp=42$ rare samples.
(c)~\emph{Amplitude amplification.} Drawn from the wavefunction flow state:
$Q/\mathrm{cost}=104$ rare samples at the same
budget, a factor $2.48$ more, which reproduces the lift of panel~(a) up to
the rounding of integer counts.
Each method is thus shown on the state its own pipeline prepares. The two
panels come from a single t-SNE fit applied to both clouds at once, so
they share a frame, and the ring densities compare directly; color encodes the
nearest mode. Untransported
residual mass---closer to the origin than to any mode---is excluded
throughout, matching the oracle's indicator; it is $3.0\%$ of the
wavefunction flow's mass, against $0.2\%$ for a finite sample of the target
itself. Tail fractions are measured
on $4\times10^{4}$ Born samples per pipeline; the Grover outcome statistics are
analytic (Methods), with no compilation, controlled state preparation, or
tensor network contraction costs included.}
\label{fig:rare-event}
\end{figure}

\section{Discussion}\label{sec:discussion}
We have addressed the exponential cost of classically simulating
wavefunction flows~\cite{bib:layden2025} using tensor
networks to compress, evolve, and represent wavefunctions. We show that this compression
converts that cost into polynomial-in-dimension growth, using the same general-purpose tensor network toolset that quantum computing already relies on for simulating quantum systems and for compiling quantum circuits~\cite{bib:berezutskii2025tnqc}. Concretely, we query a scalar-valued potential pointwise and never tabulate it on the dense $N^d$ grid. That restriction preserves the correspondence with gate-level hardware, which can only access $\Vt$ through oracle queries, because a simulator that precomputes the dense array would solve an easier problem than a circuit implementation faces (Supplementary Proposition~\ref{thm:tci-query}).

We demonstrate scaling on orthogonal Gaussian mixtures, a target family
whose density is low-rank by construction. Quantum-inspired classical
algorithms match a broad range of quantum speedups whenever the input is
low-rank and $\ell^2$-sampleable and the desired output is a classical
sample~\cite{bib:tang2019}. Both the velocity potential and the evolving
state are empirically low-rank on these targets, placing the demonstrated
instances squarely in the dequantizable regime.

Three properties of the targets explain the compression. First, the wavefunction is a smooth lift of $\sqrt{p_t}$, so
it inherits its regularity from the learned potential, and the Trotter
requirement saturates in practice. Second, learned scalar potentials resolve the target's geometry on its natural length scale and are featureless
elsewhere, and cross interpolation finds that structure without external
guidance. Third, the phase operator acts largely within the subspace the
wavefunction already occupies, which TDVP tracks variationally without
certifying it in advance.
The bond dimension $D_V$ of the learned potential is therefore the resource
parameter of the pipeline: it jointly bounds the classical query cost, the
gate count of the precompiled circuit, and the end-to-end preparation cost
(Supplementary Section~\ref{app:a7}, Propositions~\ref{thm:tci-query}--\ref{thm:tci-circuit} and Corollary~\ref{cor:total-cost}).

Our pipeline sits between quantum generative models that learn a target distribution and circuit-level simulations that assume a compiled oracle.
Existing quantum generative models train a circuit or an MPS to approximate a target~\cite{bib:han2018mps, bib:coyle2020born,
bib:zoufal2019qgan}. We instead simulate a wavefunction flow whose Born
density matches the target by construction, up to the grid, Trotter, and
compression errors.
Wavefunction flows are the transport-setting special case of the
generic Schr\"odingerisation program~\cite{bib:jinliuyu2023}, which lifts
classical linear dynamics to Schr\"odinger time-evolution. In contrast, direct
quantum-circuit simulation of the same Trotter ansatz presupposes the
precompiled phase oracle our pipeline produces, and adjacent classical
strategies~\cite{bib:araki2022adaptivepotential} exploit the same target
structure and could complement this pipeline. The pipeline is thus empirically scalable in moderate
dimensions, and produces a representation that a separate compilation procedure can consume under
explicit access and precision assumptions.

Although our tensor network simulations yield an MPS description of the qsample, they do not yet produce a compiled unitary for coherently preparing the state on quantum hardware. Three ranks determine the efficiency of classical simulation: $D_\Psi$ (state), $D_V$ (generator), and $D_O$ (readout observable). Compressing the state with a low-rank MPS pushes the difficulty onto the generator and the readout. When all three ranks stay low, the preparation-and-readout task is classically simulable.
Bond growth is only one indicator of
classical difficulty~\cite{bib:schuch2008entropy,
bib:bravyi2016clifford, bib:jozsa2008matchgates, bib:barahona1982ising}, so any advantage must identify a specific preparation-and-readout task. The task that would qualify here is readout of an observable that is efficiently oracle-accessible yet not efficiently contractible (Supplementary
Section~\ref{app:rare-oracle}). 

Compiling the prepared state is a separate problem. A compact MPS has a depth-linear sequential construction, improving to logarithmic depth with unitary circuits, which is provably optimal~\cite{bib:malz2024logdepth}, and to constant depth with adaptive circuits that use mid-circuit measurement and feedforward, at a constant-factor ancilla overhead~\cite{bib:smith2024constdepth}. Choosing among these is an open question in the wider tensor network landscape~\cite{bib:berezutskii2025tnqc}. We supply the state and generator representations that such an estimate would consume, but we do not perform circuit synthesis and resource estimation here.

As a concrete application, rare-event sampling serves as an end-to-end test of the tensor network pipeline for wavefunction flows. Amplitude amplification requires an indicator fixed in advance. When the rare events are unknown \textit{a priori}, they cannot be flagged. Guo \emph{et al.}~\cite{bib:guo2026rare} give a quantum algorithm for that setting, with optimal dependence on the rarity threshold.
Here the indicator is available because a distance rule around modes defines the tail. Rare-event probability mass survives the tensor network evolution, and on that prepared state amplitude amplification reproduces the idealized $O(1/\sqrt{p_{\mathrm{rare}}})$ scaling, yielding rare samples with $2.5\times$ fewer state preparations per accepted rare sample
than rejection sampling at the same threshold. The tail is the part of the target most easily lost to truncation, so its survival and the measured $O(1/\sqrt{p_{\mathrm{rare}}})$ scaling validate the pipeline that produced the compressed state.

\FloatBarrier
\section{Methods}\label{sec:methods}
We present the algorithms in the order of the contributions: the
representation and the Trotter splitting; the TCI$\to$TDVP pipeline and its
$V$-step compression methods; the $V$-step integrators and their
cost--accuracy tradeoffs; the pre-trained tensor network velocity that removes the runtime TCI construction, with its training; and the experimental protocol. The
variational formalism the integrators discretize is derived in Supplementary
Section~\ref{app:b4-tdvp-overview}.

\subsection{MPS representation and the scaling axis}
We represent each spatial dimension by one MPS site with physical
dimension $N$; the full $d$-dimensional grid has size $N^d$ but the
MPS itself has only $d$ sites of size $N \times D \times D$, so its
total memory is $\mathcal{O}(d N D^2)$ in place of $N^d$. Here we
study scaling primarily in $d$ at moderate $N$
(typically $N \le 32$). Scaling in $N$ is also important but is
not the focus here; we hold $N$ at moderate resolution and report
scaling in $d$.

\subsection{Trotter splitting} The continuity Hamiltonian $\Hcons = i[\Kop, \Vt]$ is split into
its kinetic part $\Kop$ (diagonal in Fourier modes) and a
time-dependent diagonal potential
$\Vt(x)$. Each Trotter step uses the eight-step commutator product
formula of Childs and Wiebe~\cite{bib:childswiebe2013}, as adopted
for wavefunction flows by Layden \emph{et al.}~\cite{bib:layden2025}
(local error
$\mathcal{O}(\Delta t^2)$, global $\mathcal{O}(K^{-1})$). $K$ acts on the dense grid via two
$d$-dimensional fast Fourier transforms (FFTs) per step; in the MPS representation it acts
site-locally because $\Kop$ is a Kronecker sum,
$\Kop = \sum_{j=1}^d I^{\otimes(j-1)} \otimes K_j^{(1)} \otimes
I^{\otimes(d-j)}$, costing
$\mathcal{O}(d N D^2 \log N)$ time (a one-dimensional FFT along the physical leg of
each core) and trivial memory. The non-trivial
ingredient is the V-step $e^{i\beta\Vt}$, which is position-diagonal but, written as an MPO, generally requires high bond dimension.

\subsection{The TCI\texorpdfstring{$\to$}{-to-}TDVP pipeline\label{app:b-pipeline}}
At a conceptual level, the two algorithm families serve complementary functions. \emph{Tensor cross interpolation}~\cite{bib:oseledets2010ttcross,bib:ttcross,bib:savostyanov2014,bib:xfac2025} constructs a low-rank matrix product operator representation of the velocity potential $\Vt$, which is the generator of the V-step rather than the unitary $e^{i\beta\Vt}$ itself. This is achieved by sampling $\Vt$ on a small, adaptively selected subset of grid points using maxvol-style pivot selection and reconstructing from those samples, without forming or evaluating $\Vt$ on the full $N^d$ grid. The \emph{time-dependent variational principle} (TDVP) subsequently evolves the wavefunction MPS through a single Trotter step under this $\Vt$-MPO by projecting the Schrödinger evolution onto the rank-$D$ MPS manifold. The one-site variant, 1TDVP, updates a single tensor at a time and keeps the bond dimension fixed; the two-site variant 2TDVP jointly updates two adjacent tensors and truncates them by singular value decomposition (SVD) back to $D_{\max}$ between them, thereby allowing adaptive bond dimension at the cost of a more expensive sweep. Between the two integrators, 1TDVP offers the better tradeoff among wall-clock time, memory, and accuracy in the regime we test (see Section~\ref{sec:results-cost}); 2TDVP is more flexible, but its adaptive bond growth is beneficial only for the most strongly entangled targets.

The V-step applies the position-diagonal unitary $e^{i\beta\Vt}$ to the
wavefunction MPS. That operator is diagonal in position, but its
bond dimension inflates with $\beta\lVert\Vt\rVert$, because a highly
oscillatory complex function is expensive to represent. The pipeline
never forms that operator. The pipeline works instead with the
\emph{generator} $\Vt$---a
real, smooth scalar field whose bond dimension stays small for the
targets we study (empirically; Supplementary Section~\ref{app:a5}) and is
generically far below that of its
exponential (Supplementary Remark~\ref{rem:tdvp-vs-tci})---and applies the exponential
only implicitly, locally, through a variational integrator. Interpolating
the generator rather than its exponential is the design choice that makes
the V-step tractable.

The pipeline has two stages. \emph{Stage 1} (TCI; Section~\ref{app:b2}).
TCI builds a low-rank, position-diagonal MPO of $\Vt$ by querying it on an
adaptively chosen set of grid points (pivots), never tabulating the dense-grid
$N^d$ field; its sample cost is $\mathcal{O}(d N D_V^2)$, exponentially
less than the $N^d$ a tabulation-first method would require. A trained MPS potential (Section~\ref{app:trained-v}) is already
an MPS and skips this stage, feeding its cores straight to Stage~2.
\emph{Stage 2} (TDVP;
Sections~\ref{app:b4}--\ref{app:se-adaptive}; the variational formalism is derived in Supplementary Section~\ref{app:b4-tdvp-overview}). The
time-dependent variational
principle~\cite{bib:haegeman2011tdvp,bib:haegeman,bib:lubich2014}
integrates $i\partial_s\ket{\Psi}=-\Vt\ket{\Psi}$ over $s\in[0,\beta]$ on
the rank-$D$ MPS manifold, whose time-$\beta$ solution is exactly
$e^{i\beta\Vt}\ket{\Psi}$ (the sign flips for the $e^{-i\beta\Vt}$
factors of the product formula). It sweeps over sites and exponentiates a small
site-local effective Hamiltonian on a Krylov subspace, so the global
unitary $e^{i\beta\Vt}$ is never assembled.

Per Trotter step the flow alternates a \emph{kinetic step}---$e^{i\alpha \Kop}$
applied axis-by-axis via one-dimensional FFTs, cheap and rank-preserving because $\Kop$
is a Kronecker sum---with the \emph{V-step} of Stage~1, repeated for $K$ Trotter
steps of the eight-letter product formula (local error
$\mathcal{O}(\Delta t^2)$, global $\mathcal{O}(K^{-1})$;
Supplementary Eq.~\eqref{eq:trotter8}). Every object touched
is of size $\mathcal{O}(d N D^2)$ (state) or $\mathcal{O}(d N D_V^2)$
(operator construction), never the dense $N^d$ grid.

\subsection{Dense-grid pseudospectral baseline\label{app:b1}}
The dense-grid
baseline applies Supplementary Eq.~\eqref{eq:trotter8} on the full $N^d$
grid: $e^{i\alpha \Kop}$ via two $d$-D FFTs (with a
$(-1)^{i_0+\dots+i_{d-1}}$ sign mask to center the kinetic
spectrum), $e^{i\beta\Vt}$ as pointwise multiplication.
Memory $16\,N^d$\,bytes (complex128); infeasible for $N^d
\gtrsim 10^9$.

\subsection{Tensor cross interpolation\label{app:b2}}
TCI~\cite{bib:oseledets2010ttcross,bib:ttcross,bib:savostyanov2014,bib:xfac2025,bib:ritter2024qtci}
constructs an MPS
approximation of a function $f : \prod_j \{1,\dots,N\} \to
\mathbb{C}$ of $d$ discrete variables without ever materializing
the dense $N^d$ tensor. The algorithm queries $f$ on a small set
of carefully chosen index tuples (\emph{pivots}) and assembles
the MPS from those queries.

\paragraph{Skeleton decomposition}
The MPS format~\cite{bib:oseledets2011} expresses $f$ as
\begin{equation}\label{eq:tt}
f(i_1, \dots, i_d) =
A_1[i_1] A_2[i_2] \cdots A_d[i_d],
\qquad A_j[i_j] \in \mathbb{C}^{D_{j-1}\times D_j},
\end{equation}
with $D_0 = D_d = 1$. TCI builds Eq.~\eqref{eq:tt} site by site
via the \emph{cross approximation}: at each bond, one selects a
set of $D_j$ row indices $I_j \subseteq
\{1,\dots,N\}^j$ and column indices
$J_j \subseteq \{1,\dots,N\}^{d-j}$, and uses the
$|I_j| \times |J_j|$ block of $f$ to build a
\emph{maxvol-optimized} interpolating skeleton
$f \approx C_j P_j^{-1} R_j$ where $P_j = f|_{I_j \times J_j}$.
A sweep in the style of the density matrix renormalization
group (DMRG) iterates between bonds, refining $I_j, J_j$
to maximize the \emph{volume} (absolute determinant) of $P_j$.

\paragraph{Application to the V-step}
The deployed pipeline applies TCI to the \emph{generator} $\Vt(x)$
itself, viewed as a real function $\{1,\dots,N\}^d \to \mathbb{R}$ on
the $N$-point grid per axis, producing a position-diagonal MPO of bond
dimension $D_V$ that TDVP then evolves the wavefunction under
(Sections~\ref{app:b4}--\ref{app:b5}). Crossing the smooth real
generator rather than the oscillatory unitary $e^{i\beta\Vt}$
keeps $D_V$ small: the exponential's rank inflates with
$\beta\lVert\Vt\rVert$, whereas the rank of $\Vt$ stays small in
practice for our targets (Supplementary Section~\ref{app:a5},
Supplementary Remark~\ref{rem:tdvp-vs-tci}).
The MPO is diagonal in the position basis at each site (physical legs
are diagonal), so each tensor $W_j[i_j, i_j']$ is nonzero only for
$i_j = i_j'$; we exploit this to reduce per-site memory by a factor of
$N$. The cross is controlled by the maxvol residual
$\| f - C_j P_j^{-1} R_j \|_\infty$, capped at $D_V$ pivots over a fixed
number of DMRG-like sweeps. We do not benchmark the
\emph{build-then-apply} alternative, which crosses $e^{i\beta\Vt}$
itself into an MPO at the inflated rank $D_V^{\exp} \ge D_V$
(Supplementary Remark~\ref{rem:tdvp-vs-tci}).

\subsection{1TDVP (one-site)\label{app:b4}}
The single-site projector gives a sweep update that, at each
site $j$, integrates two ODEs:
\begin{align}
&\text{Forward: } \partial_t A_j(t) = -i H_{\mathrm{eff},j}^{(1)}
A_j(t), \\
&\text{Backward: } \partial_t C_j(t) = +i H_{\mathrm{eff},j}^{(0)}
C_j(t),
\end{align}
where $H_{\mathrm{eff},j}^{(1)}$ acts on the $j$-th site tensor
$A_j$ contracted against the V-step MPO and the left/right
environments, and $H_{\mathrm{eff},j}^{(0)}$ acts on the bond
matrix $C_j$ between sweeps. Both are exponentiated for
time $\Delta t$ via SciPy's
\texttt{scipy.sparse.linalg.expm\_multiply}, which grows a Krylov
(Lanczos/Arnoldi) subspace adaptively until the matrix-exponential action
converges and is capped at dimension $10$; this local exponentiation is
the dominant cost. The bond dimension is exactly preserved and does not
grow, so 1TDVP is fixed-rank.

\paragraph{Cost of the fixed-rank sweep}
Because the bond dimension does not grow, the per-site Krylov
exponentiation is on an
$\mathcal{O}(N D^2)$-dimensional space; the per-Trotter-step
wall-clock time is $\mathcal{O}(d \cdot N \cdot D^3 \cdot k)$ for
$k$ Krylov iterations. On the Gaussian mixture target, this scaling is
benign: the bond dimension required to represent the
intermediate state stays well below the cap (Supplementary
Section~\ref{app:a5}), so the fixed-$D$ flow does not lose accuracy.

\subsection{2TDVP (two-site)\label{app:b5}}
2TDVP integrates a joint two-site tensor
$A_{j,j+1}(t) = A_j(t) A_{j+1}(t)$ for time $\Delta t$ via the
analogous Krylov exponentiation, then SVD-splits the result
back into single-site cores:
\begin{equation}\label{eq:tdvp2-svd}
A_{j,j+1}(t + \Delta t) = U \Sigma V^\dagger
\quad\Rightarrow\quad
A_j \leftarrow U,\ \ \ \ A_{j+1} \leftarrow \Sigma V^\dagger,
\end{equation}
truncating the singular values to a cap of $D_{\max}$ with
cutoff $\sigma_{D+1} < \mathrm{tol}$ ($\mathrm{tol}=10^{-3}$ throughout
our runs). The bond dimension thus grows
adaptively up to $D_{\max}$, capturing entanglement structure
that 1TDVP cannot. The cost per Trotter step is a factor
of $\sim N D$ larger than that of 1TDVP, with the larger Krylov
subspace and the SVD as the principal additional terms.

\paragraph{Regime of advantage over 1TDVP}
On structured-target problems where the equilibrium bond dimension is
larger than the initial-state rank (for instance, a correlated Gaussian mixture,
where the source is a single Gaussian and the target carries genuine
cross-axis covariance), the adaptive growth of 2TDVP is
essential because 1TDVP would lock at the initial bond dimension. For the
orthogonal Gaussian mixture, the required rank is smaller, and 1TDVP with
demand-driven subspace expansion (Section~\ref{app:se-adaptive}) is
competitive at a much lower cost. This is the empirical finding of
Sections~\ref{sec:results-cost} and~\ref{sec:results-bypass}.

\subsection{Subspace-expansion 1TDVP (1TDVP+SE)\label{app:se-adaptive}}
A pure two-site 2TDVP update grows the bond at every step
($\mathcal{O}(\chi^6)$); a rank-preserving 1TDVP update
($\mathcal{O}(\chi^3)$) cannot admit the perpendicular component that the
dynamics generate and stalls at fixed $\chi$. 1TDVP+SE interpolates
between the two using a subspace
expansion~\cite{bib:white2005,bib:hubig2015,bib:yangwhite2020}:
at each step, for every bond it forms the leading would-be-new singular
value of the locally evolved two-site block $C_2 + \beta\, H C_2$,
normalizes it by the dominant value, and grows the bond by one singular value
(in that perpendicular direction) iff the ratio exceeds a threshold
(default $10^{-2}$); otherwise 1TDVP evolves on the current manifold.
The bond dimension thus tracks the dynamics' demand without a full
two-site sweep at every step. On the orthogonal Gaussian mixture with the trained
MPS potential as oracle ($d{=}8$, $N{=}32$, $K \in \{40,80,120,160\}$), pure 2TDVP
self-limits its bond by truncation ($\chi^\star\!\approx\!16$,
Figure~\ref{fig:gmm-bypass}c) and reaches a lower final deviation from
the probability flow than 1TDVP+SE (Figure~\ref{fig:gmm-bypass}a),
while 1TDVP+SE is cheaper per step. This is a cost--accuracy frontier with
a modest ($\sim\!2$--$7\times$) wall-clock gap
(Figure~\ref{fig:gmm-bypass}d).

\subsection{Targets and oracles}
For Gaussian mixture targets, the optimal velocity admits a
closed-form expression via Tweedie's
identity~\cite{bib:efron2011} (Supplementary Section~\ref{app:a4}); we use this
analytic $\Vt$ to isolate tensor network compression error
from velocity learning error. For the Swiss roll target, we train a JAM scalar-potential multilayer
perceptron (MLP) oracle ($\Vt$) and then feed it into both the JAM classical baseline (via
$\nabla \Vt$) and the wavefunction flow Trotter step. For the trained MPS $V$-step study (Section~\ref{sec:results-bypass}), we instead train the
velocity potential directly as a real-valued MPS (the trained MPS potential)
whose cores feed the V-step with no runtime TCI; this learned oracle is described in Section~\ref{app:trained-v}.

\subsection{Pipeline-native training of \texorpdfstring{$\Vt$}{Vt}}
Training $\Vt$ from JAM and then plugging it into the Trotter
pipeline (the JAM-train-then-plug-in procedure just described) is correct
when the pipeline is run at the same fidelity assumed by the JAM
loss: the dense $N^d$ grid, with integration substeps fine enough
that Trotter error is negligible. Whenever the deployed pipeline
truncates the wavefunction MPS to a fixed bond cap or coarsens
the time discretization, the JAM-trained $\Vt$ no longer
minimizes sample quality loss for that pipeline.

\subsection{Pre-trained velocity-potential oracle and the trained MPS \texorpdfstring{$V$}{V}-step\label{app:trained-v}}
When $\Vt$ is itself an MPS with real-valued cores (a trained MPS potential;
Section~\ref{app:trained-v}), the unitary
$e^{\mathrm{i}\beta\Vt}$ is diagonal-unitary by construction and its
action on the wavefunction MPS is computed directly from the
$\Vt$ cores, with no TCI compression of
$e^{\mathrm{i}\beta\Vt}$ as an MPO. This removes the operator construction cost that dominates the runtime-TCI route. We
benchmark two integrators under the trained cores: 2TDVP, which drives
the scaling runs of Table~\ref{tab:scaling} together with a cold-start
bond-budget ramp on the wavefunction MPS (Section~\ref{app:b7-ramp}),
and the cheaper 1TDVP+SE, which grows the wavefunction bond on demand
and so keeps $\chi$ small (Section~\ref{app:se-adaptive}).

The deployed TCI+TDVP route (Section~\ref{app:b2}) rebuilds a fresh TCI approximation of the \emph{generator} $\Vt$ at every Trotter substep and lets
TDVP apply its exponential implicitly; a coarser \emph{build-then-apply}
alternative instead compresses $e^{\mathrm{i}\beta\Vt}$ itself into an
MPO. Either way, an operator construction step runs at every substep.
When the velocity-potential oracle is itself parameterized as a
real-valued MPS $\Vt(x) = \langle x | A_0 A_1 \cdots A_{d-1}
\rangle$ with $A_j \in \mathbb{R}^{D \times N \times D}$, even the
per-substep TCI build is unnecessary: its cores are handed to the
V-step directly. The unitary
$e^{\mathrm{i}\beta\Vt}$ is diagonal in the position basis and
exactly unitary by construction (no $|z(x)| = 1$ violation
to penalize during training). 2TDVP on the wavefunction
MPS needs only the action of $e^{\mathrm{i}\beta\Vt}$ on the
joint two-site tensor $A_{j,j+1}$, which is computable directly
from the $\Vt$ cores: the effective Hamiltonian on the
two-site block is itself block-diagonal in $(\sigma_j,
\sigma_{j+1})$, with each block a $D^2 \times D^2$ matrix built
from contractions of the $\Vt$ cores against the wavefunction's
left/right boundary tensors. Krylov-Lanczos exponentiation on
these blocks gives the two-site update in
$\mathcal{O}(d N D^3 \chi_\Psi^2)$ per V-substep, independent of
the bond dimension of $e^{\mathrm{i}\beta\Vt}$ as an MPO.

\paragraph{Eliminated operator construction cost}
The trained MPS $V$-step consumes the $\Vt$ cores and never
forms the operator as an MPO, removing the per-substep
operator construction cost of both the runtime-TCI and
build-then-apply routes. All wavefunction flow results with the trained MPS potential (Section~\ref{sec:results-bypass}, Figure~\ref{fig:gmm-bypass})
use this V-step.

\paragraph{Trained MPS potential (amortized time site)}
We parameterize $\Vt$ as a $(d{+}1)$-site real-valued MPS
(\emph{TimeSite}): sites $0,\ldots,d{-}1$ index the spatial coordinates
with local dimension $N$, and a $d$-th \emph{time site} of local dimension
$N_t{=}16$ encodes the trajectory parameter $t \in [0,1]$, with continuous
$t$ (and continuous $x$) handled by multilinear interpolation. The cores
are real (initialized as a perturbed identity), so
$e^{\mathrm{i}\beta\Vt(x,t)}$ is diagonal and exactly unitary by
construction, with no $|z|{=}1$ soft penalty, in contrast to a complex-MPS
factor model. The trained MPS potential ansatz uses one native site per spatial
axis. Training is by JAM action matching (Supplementary Eq.~\eqref{eq:jam-loss}; Adam at
$10^{-3}$, batch $256$, single-cosine schedule with patience, the bond cap
itself supplying regularization). At inference,
\texttt{get\_mps\_cores}$(t)$ returns the $d$ spatial cores at the
requested $t$ by contracting the time site against the multilinear
$t$-basis and folding the result into the last spatial core. The
oracle is trained in a frame centered on the box origin, whereas the
world grid on which the wavefunction lives places the Gaussian source at the
box center $L/2$. Before the cores enter the V-step we therefore align the
two frames by a \emph{frame roll}: each spatial core's position leg is
cyclically shifted by $N/2$ sites along that axis,
$i \mapsto (i + N/2)\bmod N$ (a periodic \texttt{roll}), so that $\Vt$ is
evaluated at the physical coordinates the state occupies.

\subsection{Action-matching (JAM) training\label{app:b6}}
We use JAM (Supplementary Eq.~\eqref{eq:jam-loss}) to learn the potential
$\Vt$ on the Swiss roll target, and as the MLP baseline for the trained MPS potential ans\"atze on the Gaussian mixture (Section~\ref{app:trained-v}).

\paragraph{Network architecture}
We parameterize
$\Vt(x; \theta) = \mathrm{MLP}_\theta\!\bigl([\phi(x), \psi(t)]
\bigr)$ where
$\phi(x) = [\sin(2\pi x/L), \cos(2\pi x/L)] \in \mathbb{R}^{2d}$
is a periodic position encoding (matching the periodic boundary
of the wavefunction grid) and $\psi(t)$ is a 32-dimensional
sinusoidal time embedding. The MLP has 3 hidden layers of width
128 with SiLU activations and a scalar output. We compute
$\nabla \Vt$ via PyTorch autograd through the network output.

\paragraph{Training procedure}
We sample $(x_0, x_1) \sim p_0 \otimes \hat p_1$, where
$\hat p_1$ is the empirical training set (50000 samples for
Gaussian mixture/Swiss roll), then sample
$t \sim \mathcal{U}[0, 1]$ and form $x_t = (1-t) x_0 + t x_1$.
Optimization uses Adam~\cite{bib:kingma2015adam} at learning rate
$10^{-3}$, batch size
256, $2 \times 10^4$ iterations, with no learning-rate
schedule. Validation is on a held-out
$10000$-sample split.

\paragraph{Use as velocity-potential oracle}
The trained $\Vt(\cdot; \theta)$ is then evaluated on the
position grid for two purposes: (i)~as the gradient flow of
the JAM classical baseline ($\dot{x}_t = \nabla
\Vt(x_t; \theta)$); (ii)~as the velocity-potential oracle plugged into
the wavefunction flow Trotter step in place of the Gaussian mixture
analytic potential. Because both consumers see the same $\Vt$, the
comparison controls for the velocity-potential oracle: the two pipelines
then differ only in how they transport it---the classical ODE
versus the Trotterized, MPS-compressed unitary flow---so the residual
gap isolates the combined tensor network contribution (V-step
compression together with the grid, Trotter, and Born-sampling error of the
wavefunction pipeline), not oracle-learning error. Throughout, \emph{JAM}
denotes this classical ODE baseline; we distinguish its
velocity source where it matters (the analytic Gaussian mixture velocity, the trained
MLP, or the trained MPS potential).

\subsection{Hyperparameter search protocol\label{app:b7}}
We Sobol-sample the hyperparameter box
$(N, K, D_{\max}) \in [N_{\min}, N_{\max}] \times [K_{\min},
K_{\max}] \times [D_{\min}, D_{\max}]$ separately at each
$d$. The exact box per $d$ appears in Section~\ref{app:b7}.
At every cell, we record SW, MMD, $\chi_{\max}$, and total
wall-clock time. We run every cell as a Modal cloud function on a
single-CPU container with 16--64\,GB memory and incremental
checkpointing.

For each $d$ we Sobol-sample
$(N, K, D_{\max})$ over the box specified in
Table~\ref{tab:search-grid}. At $d \le 4$, the box extends to
$N=128$ or $N=256$; at $d \ge 5$, container memory limits
$N \le 32$. Each cell runs $1$--$3$ seeds. Cells whose
$N^d > 10^9$ would exceed the dense-grid path's memory; for those,
we run only the MPS methods.

\begin{table}[!htbp]
\centering
\caption{\textbf{Hyperparameter search-grid maxima per $d$.}
The $N$ cap is set by container memory (dense-grid path);
the $D$ cap is the largest bond dimension we Sobol-sampled.}
\label{tab:search-grid}
\begin{tabular}{rccc}
\toprule
$d$ & $N_{\max}$ & $K_{\max}$ & $D_{\max}$ \\
\midrule
2 & 256 & 64 & 128 \\
3 &  64 & 64 &  64 \\
4 &  32 & 64 &  64 \\
5 &  32 & 64 &  32 \\
6 &  32 & 64 &  32 \\
7 &  32 & 64 &  32 \\
8 &  32 & 64 &  32 \\
\bottomrule
\end{tabular}
\end{table}

\subsection{Bond-dimension ramp schedule\label{app:b7-ramp}}
For the higher-dimensional trained MPS potential experiments, we optionally replace
the fixed wavefunction bond cap $D_{\max}$ with a stepwise schedule
$D_{\max}(k)$ that grows with the Trotter-step index $k$. The
source state $\Psi_0$ is a single-mode Gaussian on the position
grid and has $\chi = 1$; setting $D_{\max} = 64$ from $k = 0$
forces 2TDVP's SVD step to truncate a near-empty singular-value
spectrum at every Trotter step before the dynamics have grown enough
entanglement to populate it; the resulting near-degenerate singular
values make the SVD split ill-conditioned and destabilize subsequent
Trotter steps. Ramping
$D_{\max}$ in step with the natural growth of $\chi(\Psi_k)$
keeps each SVD full-rank.

\paragraph{Schedule format}
A schedule is a sorted list of $(k_{\mathrm{thresh}}, D_{\max})$
pairs; at Trotter step $k$ the current cap is the $D_{\max}$ whose
threshold is the largest $k_{\mathrm{thresh}} \leq k$. For the
orthogonal Gaussian mixture scaling runs of Table~\ref{tab:scaling}
($d \in \{8,12,16,32\}$, $K{=}160$) the cap is ramped from a small
initial value ($\chi_0 = 1$, or $\chi_0 = 3$ when the operating-point
phase swing $\beta_L = (\pi N/2L)\sqrt{d/K}$ is large enough to make a
$\chi{=}1$ start unstable) up to its final value over roughly the first
$112$ Trotter steps, then held. The exact per-$d$ schedule is recorded
with the released run configurations.

\subsection{Sample quality evaluation metrics\label{app:b8}}
Samples are drawn from the evolved wavefunction MPS by direct
conditional (perfect) sampling~\cite{bib:ferrisvidal2012}, site by
site, with no Markov chain. We use the sliced Wasserstein (SW) distance with $128$ random
projections~\cite{bib:bonneel2015} and the radial-basis-function
maximum mean discrepancy (MMD)~\cite{bib:gretton2012}; the
target--target reference of Table~\ref{tab:sw} is defined in
Section~\ref{app:b8}.

The sliced Wasserstein distance~\cite{bib:bonneel2015}
between sample sets $\{x^{(i)}\}, \{y^{(i)}\}$ projects each
sample onto $L=128$ random unit directions $\theta_\ell$ and
averages the one-dimensional Wasserstein-1 distances:
\begin{equation}\label{eq:sw}
\mathrm{SW}(\{x\}, \{y\}) = \frac{1}{L} \sum_{\ell=1}^L
W_1\!\left(\theta_\ell^\top x, \theta_\ell^\top y\right).
\end{equation}
The MMD~\cite{bib:gretton2012} uses an isotropic Gaussian
kernel; we follow standard implementations.

The target--target finite-sample reference (labeled \emph{Exact} in the
figure panels) is
$\mathrm{SW}(\{x_a\}, \{x_b\})$ for two independent draws of
size $n=500$ from the same target distribution, averaged over
$32$ pairs (independent random projections per pair). This reference is a
finite-sample scale set by $n$ and the projection count, not an
irreducible error floor. For
the Gaussian mixture, the reference decreases mildly with $d$: at our
operating point ($n=500$, $L=128$) it ranges from
$0.069$ at $d=2$ to $0.060$ at $d=8$. We compute the reference once and
report it as the target--target row of
Table~\ref{tab:sw}. The method rows in that table are evaluated at the same $n=500$
and the same $128$ projections, so the comparison is like-for-like. The
reference is regenerated by \texttt{scripts/compute\_exact\_sw\_floor.py}
in the code release.

\subsection{Rare-event sampling}
A rare event uses the nearest-mode $k$-sigma rule, with analytic tail
probability $\prare(k) = \Pr[\chi^2_d > k^2]$; Born samples closer to the domain
center than to any mode (untransported source mass) are excluded before forming
the tail fraction. Amplitude amplification chooses the Grover power
$n$ to \emph{minimize the expected cost} $(2n{+}1)/P(n)$ with
$P(n) = \sin^2((2n{+}1)\theta)$ and $\prare = \sin^2\theta$, rather than to
maximize the single-shot success probability. The latter is the textbook rule
but the wrong objective when a failed trial can be repeated: it costs
$11\%$ more on average here and makes the cost curve step for no benefit.
Writing $m = (2n{+}1)\theta$, the minimization gives $\tan m = 2m$, whose
first positive root $m^\ast = 1.1656$ yields an asymptotic cost of
$m^\ast/\sin^2 m^\ast = 1.380$ preparations per $\sqrt{\prare}$; the familiar
$\pi/2$ belongs to the maximize-$P$ rule. Since $n{=}0$ is always available at
cost $1/\prare$ (plain rejection sampling on the prepared state), amplification is
never dearer than the flow ODE pipeline on the same state. Cost is quoted in state
preparations, $(2n{+}1)$ per trial, not in oracle calls; the clean $-1/2$ power
law holds in that unit and not in the natural query conventions. Classical
rejection costs one preparation at success probability $\prare$. Rejection
therefore costs exactly $1/\prare$, so the line drawn through its markers in
Figure~\ref{fig:rare-event}(a) is an identity rather than a fit; for
amplification the fitted prefactor is a genuine test of the model, and it
passes, returning $1.383$ against the theoretical $1.380$. That fit is taken
over $\prare \le 0.2$ because the $1/\sqrt{\prare}$ law is asymptotic: beyond it the
optimal round count reaches zero, the cost flattens onto the
one-preparation floor, and the two curves meet. We simulate the
exact Grover outcome statistics rather than those of a compiled circuit, with $\prare$ taken
from the prepared state's own Born-sample tail fraction. The flow ODE pipeline
integrates $\dot x = \nabla V$ for the same learned potential with $640$
substeps, so the oracle error is common to both pipelines.

\subsection{Use of large language models}
During the preparation of this work, the authors used Claude Opus 4.7, 4.8, and 5 and
Claude Code (Anthropic) to improve language and readability, and as coding
assistants in developing the research software. All AI-assisted code was written
under the authors' direction and reviewed and tested by the authors before use in any
reported experiment. Large language models were not used to design the study,
derive the theory, or interpret the results.
After using these tools, the authors reviewed and edited the content as needed and
take full responsibility for the published article.

\FloatBarrier
\FloatBarrier
\backmatter

\section*{Data availability}
All study data are archived on Zenodo at \url{https://doi.org/10.5281/zenodo.21845216}.

\section*{Code availability}
All source code files for this paper are available at
\url{https://github.com/sygaldry-tech/tnwf-paper}.

\section*{Acknowledgments}
This study received no external funding. We thank Modal Labs, Inc.\ for in-kind
compute credits provided through its startup program, which supported the
numerical experiments; Modal Labs had no role in the design, execution, analysis,
or reporting of the study.

\section*{Author contributions}
NXK: Conceptualization, Methodology, Software, Validation, Formal analysis,
Investigation, Data curation, Visualization, Writing -- original draft.
LAW: Conceptualization, Writing -- review \& editing.
SC: Conceptualization, Software.
CR: Supervision, Funding acquisition, Writing -- review \& editing.
SV: Conceptualization, Formal analysis, Visualization, Project administration,
Writing -- review \& editing.
MJK: Conceptualization, Resources, Supervision, Project administration,
Writing -- review \& editing.
All authors read and approved the final manuscript.

\section*{Competing interests}
All authors are employees of Sygaldry Technologies, Inc., which develops
algorithms and hardware for quantum-accelerated artificial intelligence. Some
authors hold equity or equity-based compensation in the company, and C.R.\ and
M.J.K.\ are co-founders.

%
\clearpage
\setcounter{section}{0}
\setcounter{figure}{0}
\setcounter{table}{0}
\setcounter{equation}{0}
\renewcommand{\thesection}{S\arabic{section}}
\renewcommand{\thesubsection}{S\arabic{section}.\arabic{subsection}}
\renewcommand{\thefigure}{S\arabic{figure}}
\renewcommand{\thetable}{S\arabic{table}}
\renewcommand{\theequation}{S\arabic{equation}}

\begingroup
\centering
{\LARGE\bfseries Supplementary Information\par}
\vspace{0.7em}
{\large for ``Scalable quantum simulation of continuous-time generative models\\
via tensor networks''\par}
\par\endgroup
\vspace{1.2em}
\noindent\rule{\textwidth}{0.4pt}
\vspace{0.9em}

\noindent{\small This appendix is the Supplementary Information of the journal
version, included here so that the arXiv posting is self-contained. Section,
figure, table and equation numbers are unchanged from that version.}
\bigskip


\section{Theory}\label{app:theory}

Here we follow the same argument as in the main text.
In Section~\ref{app:a1}, we set up the wavefunction flow construction and
the qsample / amplitude estimation quadratic advantage that motivates
the work. In Sections~\ref{app:a5} and~\ref{app:a7} we identify the resource
behind the acceleration, namely the bond dimension of the
velocity potential and the resulting query- and gate complexity bounds, set
against the dense $N^d$ cost of Section~\ref{app:a6}. In between, we collect the
supporting derivations: the Trotter product formula
(Section~\ref{app:a2}), action matching (Section~\ref{app:a3}), and
the analytic Gaussian mixture velocity (Section~\ref{app:a4}).

\subsection{Wavefunction flows fundamentals}\label{app:a1}
In this section, we collect the construction of wavefunction flows
from~\cite{bib:layden2025} in the form needed by the body of the
paper and connect it to amplitude estimation primitives.

\paragraph{Continuity equation and conservative velocities}
Let $p_0, p_1$ be probability densities on a periodic spatial
domain $[0,L)^d$. The transport problem is to find a time-dependent
velocity field $v_t : [0,L)^d \to \mathbb{R}^d$ and a marginal
density $p_t$ such that
\begin{equation}\label{eq:continuity}
\partial_t p_t + \nabla \cdot (p_t \, v_t) = 0,
\qquad p_{t=0} = p_0, \quad p_{t=T} = p_1.
\end{equation}
A \emph{stochastic interpolant}~\cite{bib:interpolants} of the form
$x_t = a(t) x_0 + b(t) x_1$ with boundary values
$a(0)=1, a(T)=0, b(0)=0, b(T)=1$ produces an
explicit family of $p_t$. (We reserve $\alpha, \beta$ for the
Trotter substep angles of Section~\ref{app:a2}.) The associated
\emph{conditional mean velocity}
\begin{equation}\label{eq:cond-mean}
v_t(x) = \mathbb{E}\bigl[\dot{a}(t)\, x_0
   + \dot{b}(t)\, x_1 \;\bigm|\; x_t = x\bigr],
\end{equation}
satisfies Eq.~\eqref{eq:continuity} and transports a sample
$x_0 \sim p_0$ to a sample $x_T \sim p_1$ along the
ODE $\dot{x}_t = v_t(x_t)$. We focus on
\emph{conservative} parameterizations $v_t = \nabla \Vt$, which
admit a quantum lift. In Section~\ref{app:a3} we describe the variational
characterization of the optimal $\Vt$.

\paragraph{Lift to a wavefunction}
Discretize $[0,L)^d$ on a uniform grid of $N^d$ points
$\{x_k\}$ with cell volume $\Delta x^d = (L/N)^d$, and write a complex
amplitude $\psit \in \mathbb{C}^{N^d}$ on this grid, normalized as a
discrete state, $\sum_k |\psit(x_k)|^2 = 1$. A continuum density $p$ maps
to grid amplitudes by $\Psi(x_k) = \sqrt{p(x_k)\,\Delta x^d}$ (so that
$\sum_k|\Psi(x_k)|^2 \approx \int p\,dx = 1$); the cell-volume factor
$\Delta x^d$ ensures that the discrete Born weights approximate the
continuum probabilities, and expectations $\langle\Psi|F|\Psi\rangle$ are
correct without further rescaling. The wavefunction flow
construction~\cite{bib:layden2025} promotes the conservative
velocity field $v_t = \nabla \Vt$ to a unitary evolution
\begin{equation}\label{eq:psi-evolution}
i \partial_t \psit = \Hcons \psit,
\qquad \Hcons = i[\Kop, \Vt],
\end{equation}
where $\Kop$ is the discrete kinetic operator (with periodic boundary
conditions, diagonalized by the discrete Fourier transform (DFT)
$\mathcal{F}$ as $\Kop = \mathcal{F}^\dagger \, \tfrac{1}{2}|p|^2 \,
\mathcal{F}$ on the grid momenta $p$; the Fourier pseudospectral
discretization of Refs.~\cite{bib:layden2025,bib:childs2022realspace}), and
$\Vt = \mathrm{diag}(\Vt(x_k))$ is the position-diagonal potential.
The commutator structure $i[\Kop,\Vt]$ is anti-Hermitian times $i$,
hence Hermitian, so $\Hcons$ generates a unitary; because $\Vt$
(and hence $\Hcons$) depends on $t$, this is the time-ordered
evolution $U(T) = \mathcal{T}\exp(-i\int_0^T \Hcons\,dt)$
(the frozen-coefficient $e^{-iT\Hcons}$ only when $\Hcons$ is held
fixed over a step). This $\Hcons$ is the conservative form of
the \emph{continuity Hamiltonian} of Ref.~\cite{bib:layden2025},
the discretization of
$\hat{H}_t = \tfrac{1}{2}[\hat{p} \cdot v_t(\hat{x}) +
v_t(\hat{x}) \cdot \hat{p}]$ for general velocity fields.

\paragraph{Comparison with the Madelung formulation}
A different connection between Schrödinger evolution and the
continuity equation has been known since Madelung: if $\Psi_t$
evolves under a canonical kinetic-plus-potential Hamiltonian
$\widehat{K} + V_t(\widehat{x})$, then $|\Psi_t|^2$ obeys a
continuity equation, but with a velocity field that depends on
the phase of $\Psi_t$. The phase, in turn, evolves through the
density via a quantum-potential term, so the two equations are
coupled. The wavefunction flow Hamiltonian
$\Hcons = i[\Kop, \Vt]$ decouples them by construction: the
velocity field is specified directly by the learned $\Vt$, the
phase $\arg\Psi_t$ remains trivial along the trajectory, and the
density marginal satisfies the continuity
equation~\eqref{eq:psi-continuity} without a quantum-potential
correction. The price is that $\Hcons$ is not of canonical
kinetic-plus-potential form and would not arise in nature.
This decoupling renders the V-step a
position-diagonal phase rather than a self-coupled state evolution,
and hence amenable to the MPS-compression treatment of
the rest of the paper.

\paragraph{Correctness of the transport}
A direct derivation~\cite{bib:layden2025} shows that, in the
continuum, if the initial wavefunction is chosen with phase
$\arg \Psi_0(x) = 0$ and amplitude $|\Psi_0(x)|^2 = p_0(x)$,
then the marginal probability of the evolved wavefunction
satisfies
\begin{equation}\label{eq:psi-continuity}
\partial_t |\psit|^2 = -\nabla \cdot (|\psit|^2 \, v_t),
\qquad v_t = \nabla \Vt,
\end{equation}
matching Eq.~\eqref{eq:continuity} as an exact identity for the
$\Hcons$-evolution. The digital simulator then incurs
spatial-discretization and Trotter/product-formula errors, so
$|\Psi_T(x)|^2 = p_1(x)$ up to these errors.

\paragraph{The qsample and amplitude estimation}
For any bounded function $f : [0,L)^d \to \mathbb{R}$ (rescaled to
$[0,1]$, as amplitude estimation requires), the expectation
$\mathbb{E}_{x \sim p_1}[f(x)] = \langle \Psi_T | F |
\Psi_T \rangle$ where $F = \mathrm{diag}(f(x_k))$ is its
position-diagonal embedding. Given a unitary
$U_{\mathrm{prep}}$ that prepares $|\Psi_T\rangle$ from a
fiducial state, together with an oracle that rotates an ancilla by
$f$ (so that $\langle\Psi_T|F|\Psi_T\rangle$ becomes an amplitude),
amplitude
estimation~\cite{bib:bhmt,bib:montanaro2015} returns an
$\epsilon$-accurate estimate of
$\langle \Psi_T | F | \Psi_T \rangle$ using
$\mathcal{O}(\epsilon^{-1})$ queries to $U_{\mathrm{prep}}$ and that oracle,
versus the $\mathcal{O}(\epsilon^{-2})$ classical samples a
particle ensemble of size $\mathcal{O}(\epsilon^{-2})$ would
require. This quadratic sampling advantage is the principal
intended application of wavefunction flows. We neither compile amplitude
estimation to the circuit level nor demonstrate it here. What we do characterize
is the cost of compiling $U(T)$ at the dimensions where the advantage becomes
interesting, together with the complementary primitive the prepared state does
support---amplitude amplification for drawing rare samples rather than
estimating their probability (Section~\ref{sec:results-rare}, Figure~\ref{fig:rare-event}).

\subsection{Continuity Hamiltonian and Trotter
splitting}\label{app:a2}
The continuity Hamiltonian $\Hcons = i[\Kop,\Vt]$ does not split
into a sum of locally simulable terms directly. Its commutator
structure does admit an exact product formula in $K$ and
$\Vt$ separately, both of which split locally:
$\Kop$ is separable in coordinates and exponentiates via one-dimensional FFTs, and
$\Vt$ is diagonal in position.

\paragraph{Eight-letter product formula}
The starting point is the group-commutator product formula of
Childs and Wiebe~\cite{bib:childswiebe2013}, the Baker--Campbell--Hausdorff (BCH)-like identity
$e^{-i \Delta t \, i[A,B]} = e^{i\beta B} e^{i\alpha A}
e^{-i\beta B} e^{-i\alpha A} \cdot
e^{-i\beta B} e^{-i\alpha A} e^{i\beta B} e^{i\alpha A}
+ \mathcal{O}(\Delta t^2)$ for any operators $A, B$ with
$\alpha\beta = \Delta t / 2$. The right-hand side is a group commutator
times its sign-flipped copy, $C(\beta,\alpha)\,C(-\beta,-\alpha)$ with
$C(\beta,\alpha) = e^{i\beta B} e^{i\alpha A} e^{-i\beta B} e^{-i\alpha A}$.
Because the substep angles scale as $\alpha,\beta \sim \sqrt{\Delta t}$
(fixed later in this section), the corrections are graded by total degree in
$(\alpha,\beta)$, with $\alpha\beta \sim \Delta t$. The leading order of
the product is $\Delta t \cdot [A,B]$. The degree-three corrections
($\alpha^2\beta$, $\alpha\beta^2 \sim \Delta t^{3/2}$) cancel by the
symmetric $C(\beta,\alpha)\,C(-\beta,-\alpha)$ structure, and the
surviving degree-four terms give local error $\mathcal{O}(\Delta t^2)$,
consistent with Eq.~(38) of Ref.~\cite{bib:layden2025}. Specializing to
$A = \Kop, B = \Vt$:
\begin{equation}\label{eq:trotter8}
W(\alpha,\beta) = e^{i\beta\Vt} e^{i\alpha K} e^{-i\beta\Vt}
e^{-i\alpha K} e^{-i\beta\Vt} e^{-i\alpha K} e^{i\beta\Vt}
e^{i\alpha K},
\end{equation}
which approximates $e^{-i\Delta t \, \Hcons}$ to
$\mathcal{O}(\Delta t^2)$ per step, hence global error
$\mathcal{O}(K^{-1})$ over $K = T/\Delta t$ steps. Reaching the
$\mathcal{O}(K^{-2})$ of an ordinary symmetric Trotter formula would
require a higher-order commutator construction; the group-commutator
structure forced by the $i[\Kop,\Vt]$ generator caps this formula at first
order.

\paragraph{Choosing \texorpdfstring{$\alpha$}{alpha} and \texorpdfstring{$\beta$}{beta}}
Setting $\alpha = (L / (\pi N)) \sqrt{\Delta t / d}$ and
$\beta = (\pi N / (2L)) \sqrt{d \Delta t}$ ensures
$\alpha \beta = \Delta t / 2$ and balances the kinetic and
potential phase scales: $\alpha \Kop$ has $\mathcal{O}(\sqrt{\Delta t})$
spectral spread on the discrete grid, and $\beta \Vt$ has the
matching $\mathcal{O}(\sqrt{\Delta t})$ phase swing. With this
balance, the product formula is well-conditioned across the
parameter ranges we test.

\paragraph{Cost decomposition per Trotter step}
\begin{itemize}
  \item \textbf{Kinetic step} $e^{i\alpha \Kop}$: $\Kop$ is the Kronecker
    sum of single-site kinetic operators,
    $\Kop = \sum_{j=1}^d I^{\otimes(j-1)} \otimes K_j^{(1)} \otimes
    I^{\otimes(d-j)}$, each
    diagonalized by the one-dimensional DFT, so $e^{i\alpha K} = \bigotimes_j
    e^{i\alpha K_j^{(1)}}$ acts site-locally.
    Cost: $\mathcal{O}(d N D^2 \log N)$ on an MPS, $\mathcal{O}(N^d \log N)$ on
    the dense grid.
  \item \textbf{Position-diagonal V-step} $e^{i\beta \Vt(x)}$: the
    operator is diagonal in the position basis but
    generically of high bond dimension as an MPO, because $\Vt(x)$
    depends jointly on all coordinates. This is the bottleneck we
    address in the rest of the paper. The V-step compression
    strategies (the Methods) all target
    this single subroutine.
\end{itemize}
Per Trotter step there are four $K$ steps and four $\Vt$ steps
(Eq.~\eqref{eq:trotter8}). We report wall-clock time per
Trotter step, i.e.\ per application of $W$.

\paragraph{Rank-preservation under tensor-product operators}
The kinetic step is benign for an MPS representation in a stronger
sense than just site-locality: it preserves bond dimensions
exactly. Let
$U = U^{(1)} \otimes U^{(2)} \otimes \cdots \otimes U^{(d)}$ with
each $U^{(j)} \in \mathbb{C}^{N \times N}$, and let $\Psi$ be an
MPS with cores $\{G^{(j)}\}$ and bond dimensions $\{r_j\}$. Then $U\Psi$
admits an MPS representation with the \emph{same} bond dimensions
$\{r_j\}$, obtained by replacing each core in place,
\begin{equation}\label{eq:rank-preserving-update}
\widetilde{G}^{(j)}[:,\,i_j,\,:] = \sum_{m=1}^N
U^{(j)}_{i_j m}\, G^{(j)}[:,\,m,\,:]\,,
\end{equation}
because $U^{(j)}$ acts only on the physical leg of core $j$ and
leaves the bond legs untouched. Applied to
$U = e^{i\alpha K} = \bigotimes_j e^{i\alpha K_j^{(1)}}$, this is
the formal reason the kinetic Trotter substep is memory-benign
on every V-step compression strategy in
the Methods, regardless of whether bond dimensions are
otherwise being adapted.

\subsection{Time-dependent variational principle (TDVP)}\label{app:b4-tdvp-overview}
TDVP~\cite{bib:haegeman2011tdvp,bib:haegeman,bib:lubich2014,bib:paeckel2019}
formulates
$\partial_t |\Psi\rangle = -i H |\Psi\rangle$ as the
best approximation on the rank-$D$ MPS manifold
$\mathcal{M}_D$:
\begin{equation}\label{eq:tdvp}
\partial_t |\Psi(t)\rangle = -i\, \mathcal{P}_{T_{\Psi(t)}\mathcal{M}_D}
H |\Psi(t)\rangle,
\end{equation}
where $\mathcal{P}_{T_{\Psi}\mathcal{M}_D}$ projects onto the
tangent space of $\mathcal{M}_D$ at $\Psi$. Putting the MPS in the
left/right canonical gauge (each core left- or right-orthonormal up to
the active bond) splits this projector into a sum of mutually
orthogonal per-site (one-site) or per-bond (two-site) terms; applying
them in sequence yields the site-by-site Lie--Trotter sweep. The projector ensures that the
flow stays exactly on $\mathcal{M}_D$, so the bond dimension is
preserved (one-site) or grows under SVD truncation (two-site).


\subsection{Action matching and the JAM loss}\label{app:a3}
\emph{Action matching} (AM) and its joint variant JAM
\cite{bib:jam2023} are the velocity learning algorithms underlying
wavefunction flows. The construction begins from
the optimization characterization of the conservative velocity
$v_t = \nabla \Vt$ that solves the continuity
equation~\eqref{eq:continuity} for given $p_0, p_1$.

\paragraph{Variational characterization}
Among all conservative velocity fields, the one that satisfies
Eq.~\eqref{eq:continuity} for the marginals
$\{p_t\}_{t\in[0,T]}$ of a stochastic interpolant minimizes
the kinetic action
\begin{equation}\label{eq:action}
\mathcal{A}[\Vt] = \frac{1}{2}\int_0^T \int |\nabla \Vt(x)|^2
p_t(x)\, dx\, dt
\end{equation}
subject to the continuity equation~\eqref{eq:continuity}
(a Benamou--Brenier-type characterization), with $\Vt$ determined
up to an additive constant. Direct minimization of
Eq.~\eqref{eq:action} requires evaluating $p_t$, which is
intractable. AM avoids this by integrating by parts.

\paragraph{The AM loss}
Differentiating Eq.~\eqref{eq:action} along the gradient flow of
$\Vt$, using Eq.~\eqref{eq:continuity} and integration by parts,
and writing
$x_t = a(t) x_0 + b(t) x_1$, gives
\begin{equation}\label{eq:am-loss}
\Vt^{\star} = \arg\min_\Vt \,
\mathbb{E}_{t \sim \mathcal{U}[0,T]}\,
\mathbb{E}_{x_0, x_1}
\left[\, \tfrac{1}{2} \|\nabla \Vt(x_t)\|^2 -
\nabla \Vt(x_t) \cdot \bigl(\dot a(t) x_0
+ \dot b(t) x_1\bigr) \right].
\end{equation}
The minimizer is exactly the conditional-mean
velocity~\eqref{eq:cond-mean}, and the loss is computable purely
from interpolant draws $(x_0, x_1) \to x_t$, with no density
evaluation needed. JAM is the practical regression-style
relaxation
\begin{equation}\label{eq:jam-loss}
\Vt^{\star} \approx \arg\min_\Vt
\mathbb{E}_t \mathbb{E}_{x_0, x_1}
\bigl[ \|\nabla \Vt(x_t) - (\dot a(t) x_0 + \dot b(t) x_1)\|^2
\bigr],
\end{equation}
identical to Eq.~\eqref{eq:am-loss} up to a $\Vt$-independent
constant. Completing the square in Eq.~\eqref{eq:am-loss} also
shows that $\mathcal{L}_{\mathrm{JAM}}$ is, up to a
$\Vt$-independent constant, equal to
$\frac{1}{2}\,\mathbb{E}\!\bigl[\|\nabla \Vt(x_t) - u_t^\star\|^2
\bigr]$, where $u_t^\star$ is the optimal conditional-mean
velocity~\eqref{eq:cond-mean}. Restricting to conservative
velocities is therefore not a loss of generality \emph{for the
specific setting considered here}---a Gaussian source with a linear
stochastic interpolant, whose optimal conditional-mean velocity is a
score, hence a gradient---because the conservative class then regresses
to the same optimal target as unconstrained flow matching. This is not a
claim about arbitrary probability flows, whose optimal velocity need not
be conservative. We use Eq.~\eqref{eq:jam-loss} as the training loss.

\subsection{Exact velocity for Gaussian mixture targets}\label{app:a4}
For Gaussian mixture targets, the conditional mean velocity is
known analytically; we use this exact $\Vt$ throughout
to isolate the compression error from velocity learning.

\paragraph{Tweedie's identity}
Let $x_t = x_0 + \sigma_t \,\xi$ with
$\xi \sim \mathcal{N}(0, I)$ a Gaussian shift of width $\sigma_t$
about an arbitrary $x_0$-distribution. The marginal density
$p_t(x_t) = \int p_0(x_0) \mathcal{N}(x_t; x_0, \sigma_t^2 I)\,dx_0$
satisfies Tweedie's identity~\cite{bib:efron2011}:
\begin{equation}\label{eq:tweedie}
\mathbb{E}[x_0 \mid x_t = x] = x + \sigma_t^2 \nabla_x \log p_t(x).
\end{equation}
Equivalently, the score $\nabla_x \log p_t$ is the
denoising direction. Eq.~\eqref{eq:tweedie} is exact, not
asymptotic, for any prior $p_0$ and any Gaussian shift width.

\paragraph{Specializing to a linear interpolant}
Let $x_t = (1-t) x_0 + t x_1$ with
$x_0 \sim \mathcal{N}(\mu_0, \sigma_0^2 I)$ and $x_1 \sim p_1$, and
let $\sigma_t^2 = (1-t)^2 \sigma_0^2 + t^2 \sigma_1^2$ be the
(effective) variance contributed by the source after rescaling
into the $x_t$ frame. The conditional mean
velocity~\eqref{eq:cond-mean} for this interpolant is
\begin{equation}
v_t(x) = \mathbb{E}\bigl[x_1 - x_0 \mid x_t = x\bigr]
= \frac{x - \mathbb{E}[x_0 \mid x_t = x]}{t}
= \frac{x - \mu_0}{t} +
\frac{(1-t)\sigma_0^2}{t}\, \nabla_x \log p_t(x),
\end{equation}
applying Eq.~\eqref{eq:tweedie} in the second equality
to the Gaussian-noise channel
$x_t - (1-t)\mu_0 = t x_1 + (1-t)(x_0 - \mu_0)$, which gives
$\mathbb{E}[x_0\mid x_t=x] = \mu_0 - (1-t)\sigma_0^2 \nabla_x
\log p_t(x)$. Because this $v_t$ is the gradient of a
scalar potential, we may integrate to recover $\Vt$ up to a
$x$-independent additive constant:
\begin{equation}\label{eq:Vt-gmm}
\Vt(x) = \frac{\|x - \mu_0\|^2}{2 t}
+ \frac{(1-t)\sigma_0^2}{t} \log p_t(x).
\end{equation}

\paragraph{Closed-form for orthogonal-mode mixtures}
If the target $p_1$ is a Gaussian mixture
$p_1(x) = \tfrac{1}{M}\sum_m \mathcal{N}(x; \mu_m, \sigma_1^2 I)$
with $M$ centers $\{\mu_m\}$, then the marginal of the linear
interpolant---the convolution of the rescaled target with the
rescaled Gaussian source---is again a Gaussian mixture (a
convolution of Gaussians is a Gaussian):
\begin{equation}\label{eq:p_t-gmm}
p_t(x) = \frac{1}{M} \sum_m
\mathcal{N}\bigl(x;\, (1-t)\mu_0 + t \mu_m,\, \sigma_t^2 I\bigr).
\end{equation}
Substituting Eq.~\eqref{eq:p_t-gmm} into
Eq.~\eqref{eq:Vt-gmm} produces an
$\Vt(x)$ involving only an explicit log-sum-exp over the $M$
shifted centers. For the orthogonal-mode
Gaussian mixtures we use ($2d$ centers $\{\pm c\,e_j\}_{j=1}^d$, with
isotropic covariance), each coordinate axis carries an
independent two-cluster mixture in the score. Still, the joint
$\log p_t$ remains a single coupled log-sum-exp because
$p_t$ does not factorize over coordinates. This
non-factorization gives the position-diagonal
$\Vt$ a generically high bond dimension in the MPO,
thus motivates the compression approach.

\paragraph{Use in the experiments}
We evaluate Eq.~\eqref{eq:Vt-gmm} on the
$N^d$ grid points to construct the dense-grid
$\mathrm{diag}(\Vt(x_k))$ for the dense-grid baseline, and we hand the
same $\Vt(x)$ as a callable to TCI for MPO construction. The analytic $\Vt$ isolates the compression error from velocity
learning. We use the JAM-trained oracle (MLP, or the trained MPS potential of
Methods, Section~\ref{app:trained-v}) for the Swiss roll target and for the trained MPS potential study (Section~\ref{sec:results-bypass}).

\subsection{Bond dimension of orthogonal-mode Gaussian mixtures}
\label{app:a5}
The orthogonal-mode Gaussian mixture target with $2d$ centers
$\{\pm c\, e_j\}_{j=1}^d$ admits an MPS decomposition of bond
dimension at most $2d$. The argument is direct: the marginal $p_t$ of
Eq.~\eqref{eq:p_t-gmm} is a uniform mixture of $2d$ shifted
Gaussians, and each shifted Gaussian factorizes across coordinates
because the covariance $\sigma_t^2 I$ is isotropic. The
$m$-th component
$\mathcal{N}\!\bigl(x;(1-t)\mu_0 + t\mu_m,\sigma_t^2 I\bigr) =
\prod_{j=1}^d \mathcal{N}\!\bigl(x_j; \cdot,\sigma_t^2\bigr)$
is a tensor product of single-coordinate Gaussians, hence an MPS
of bond dimension exactly one. The full $2d$-component mixture
therefore admits an MPS representation of bond dimension at most $2d$:
one rank-one component per center, summed across the bond
dimension.
This $\mathcal{O}(d)$ bound applies to $p_t$ as a density
tensor. The potential $\Vt$ in Eq.~\eqref{eq:Vt-gmm} depends on
$p_t$ through $\log p_t$, and the V-step phase $e^{i\beta\Vt}$
depends on it further through complex exponentiation; neither
operation preserves bond dimension in general, so a sum of $2d$ rank-one
product Gaussians need not retain $\mathcal{O}(d)$ approximate
bond dimension after $\log$ or $\exp$. We therefore treat the achieved
bond dimensions of $\Vt$ and $e^{i\beta\Vt}$ as quantities to be measured
empirically rather than derived from the rank-$2d$ bound on
$p_t$. That the relevant ranks remain compressible for the
targets in this paper is an experimental finding, not a corollary
of the Gaussian mixture closed form.

In our experiments, the empirical chi-trajectories of 2TDVP
confirm that the realized bond dimension of the propagated MPS
$\psit$ stabilizes well below $2d$
at small times, growing as $t \to 1$ to a saturating value
compatible with $D \in \{16, 32\}$ for $d \le 8$. Because the
rank-$2d$ bound is structural only for $p_t$ itself and not for
the wavefunction $\sqrt{p_t}$, the V-step phase, or the propagated
$\psit$, we read this as empirical evidence that the operating
rank cap is adequate for these targets rather than a worst-case
guarantee.

\subsection{Cost models}\label{app:a6}
Per Trotter step, the dense-grid baseline costs
$\mathcal{O}(N^d \log N)$ in time and $\mathcal{O}(N^d)$ in
memory, dominated by the $d$-dimensional FFT. The MPS V-step
is $\mathcal{O}(d N D^3)$ for a 1TDVP sweep and
$\mathcal{O}(d N D^2 m)$ for a TCI sweep with $m$ pivots, with
memory $\mathcal{O}(d N D^2)$ throughout. If $D$ scales linearly
with $d$ (an empirical assumption for our targets, \emph{not} implied by the
rank-$2d$ bound on $p_t$ of Section~\ref{app:a5}, which does not pass
through $\log$ and $\exp$ to $\Vt$), the 1TDVP wall-clock time scales as $d^4$,
a polynomial cost contrasting with the dense $N^d$. Empirically
(Section~\ref{sec:results-cost}) we measure
$d^{2.7}$ for fixed-$D$ runs, where the $D^3$ contribution is
absorbed into the prefactor. In Section~\ref{app:a7} we discuss how
the observed tensor network costs depend on the achieved
bond dimensions of the queried V-step objects and of the maintained
MPS state, and where these sit relative to the analytic
smoothness conditions of~\cite{bib:layden2025}.

\subsection{The bond dimension of \texorpdfstring{$V_t$}{Vt} as the unified resource
parameter}\label{app:a7}

The Trotter-step count and grid resolution bounds of
Layden \emph{et al.}~\cite{bib:layden2025} (Theorems~1 and~2
therein) fix the number of Trotter steps
$r = \mathcal{O}(d^{6}T^{2+2/s}/\epsilon^{1+2/s})$ and the
grid resolution $N = \mathcal{O}(L d (T/\epsilon)^{1/(2s)})$
required to prepare $\ket{\Psi_T}$ to error $\epsilon$ (here
and throughout, $s$ is the adjustable smoothness parameter of
Ref.~\cite{bib:layden2025}: an integer $s \ge (d+7)/4$ for which
$\sqrt{p_t}$ and $V_t$ are assumed $2(s{+}1)$-times continuously
differentiable, so larger $s$ can be chosen for smoother
targets). Those bounds charge the V-step $e^{i\beta\Vt}$ to an oracle: on
hardware, $\Vt$ is computed into an ancilla register and uncomputed around
a phase rotation~\cite{bib:layden2025}, at a cost set by the
reversible-arithmetic circuit for the learned $\Vt$ (for instance, an MLP
forward pass), for which no a priori certificate is available. Absent
structural access of this kind---given only pointwise evaluation of a
generic $\Vt$ coupling all $d$ coordinates---both the classical simulation
cost and the gate count of a single V-step scale as $\Theta(N^d)$ in the
worst case.

We show in this section that a single classical-algebraic property
of $\Vt$, its bond dimension $D_V$, simultaneously governs (i)~the
number of classical oracle queries to $\Vt$ required to build the
V-step by TCI (Proposition~\ref{thm:tci-query}); (ii)~the gate count of
a precompiled quantum circuit implementing the V-step
(Proposition~\ref{thm:tci-circuit}); and (iii)~the empirically measured
wall-clock time and memory of TDVP (Section~\ref{sec:results-cost}). The bond
dimension is thus a shared structural parameter, though the quantum
implementation still requires separate analyses of data access, arithmetic, control,
and error.
Corollary~\ref{cor:total-cost} composes the
per-step bounds with the Trotter-step count of
Ref.~\cite{bib:layden2025} to obtain end-to-end resource
estimates for qsample preparation, and
Corollary~\ref{cor:gmm-precompile} specializes to the
orthogonal-mode Gaussian mixture target under the assumption
$D_V = \mathcal{O}(d)$, motivated by the mixture structure of
Section~\ref{app:a5} and supported empirically
(Section~\ref{sec:results-bypass}). The remarks that follow
discuss why TDVP can achieve a smaller constant than a
build-then-apply pipeline, and where the
bond-dimension assumption sits relative to the analytic smoothness
conditions of~\cite{bib:layden2025}.

\paragraph{Setup and definitions}
Fix $t \in [0,T]$ and view $\Vt$ as a $d$-index tensor
$\Vt(i_1,\dots,i_d)$ on the $N$-per-axis grid $X_N^d$.

\begin{definition}[Bond dimension]\label{def:tt-rank}
The tensor $\Vt$ has \emph{bond dimension} $D_V$ (the \emph{tensor-train (TT)
rank} of the numerical analysis literature) if, for every
$\ell \in \{1,\dots,d-1\}$, the $\ell$-th matrix unfolding of
$\Vt$, the $N^{\ell} \times N^{d-\ell}$ matrix obtained by
grouping the first $\ell$ indices as rows and the remaining
$d-\ell$ as columns, has rank at most $D_V$. Equivalently,
$\Vt$ admits an MPS representation of the form of Eq.~(1) of the main text with
all bond dimensions bounded by $D_V$.
\end{definition}

Smoothness conditions of the kind assumed in~\cite{bib:layden2025}
control spectral discretization error but do not, in themselves,
imply a small bond dimension (Remark~\ref{rem:beyond-smoothness}).
MPS-compressibility is therefore an additional structural
assumption on the target geometry; we verify it analytically for
the mixture density of orthogonal-mode Gaussian mixtures in
Section~\ref{app:a5} and empirically for
the trained MPS potential on the Gaussian mixture in
Section~\ref{sec:results-bypass}.

\paragraph{Per-step query and gate complexity}
Our first result bounds the cost of building the V-step generator by
tensor cross interpolation; the second compiles the V-step it
defines into a quantum circuit.

\begin{proposition}[Query complexity of the TCI V-step build]
\label{thm:tci-query}
Let $\Vt$ have bond dimension at most $D_V$ on $X_N^d$. Tensor
cross interpolation~\cite{bib:oseledets2010ttcross,bib:ttcross,bib:xfac2025}
reconstructs $\Vt$ exactly, as a rank-$D_V$ MPS
$\widetilde{V}_t$ with $\widetilde{V}_t = \Vt$ on the full grid, using
\begin{equation}\label{eq:query-perstep}
Q_{\mathrm{step}}
\;=\;
\mathcal{O}\!\left(d\,N\,D_V^{\,2}\right)
\end{equation}
pointwise evaluations of $\Vt$; the diagonal V-step unitary
$e^{i\beta\widetilde{V}_t}$ it defines then equals $e^{i\beta\Vt}$
exactly, for every $\beta$.
\end{proposition}

\begin{proof}[Proof sketch]
The maxvol pivot-selection scheme sweeps over the $d-1$ bonds; at
each bond it evaluates a $D_V \times N \times D_V$ block of the
target tensor at maxvol-optimal pivots, i.e.\
$\mathcal{O}(N D_V^{\,2})$ evaluations of $\Vt$ per bond and
$\mathcal{O}(d N D_V^{\,2})$ in total. When every unfolding has rank
at most $D_V$ and the pivot submatrices are nonsingular, the
skeleton (cross) decomposition is exact~\cite{bib:oseledets2010ttcross}.
When $\Vt$ is only approximately of bond dimension
$D_V$---the practical case---maxvol cross interpolation is
quasi-optimal in the sup norm~\cite{bib:savostyanov2014}, and the
elementwise error propagates to the V-step through the Lipschitz
bound $\|e^{i\beta\widetilde{\Vt}} - e^{i\beta\Vt}\|_\infty
\le |\beta|\,\|\widetilde{\Vt} - \Vt\|_\infty$, so a per-step
operator-norm tolerance $\eta_V$ requires elementwise accuracy
$\eta_V/|\beta|$ on the generator; the number of refinement sweeps
this takes is set by the singular-value profile of the unfoldings
and is small in practice~\cite{bib:xfac2025}, though no general
sweep-count guarantee is available.
\end{proof}

\begin{proposition}[Quantum gate complexity of the V-step]
\label{thm:tci-circuit}
Let $\widetilde{V}_t$ be a rank-$D_V$ MPS on $X_N^d$
(Proposition~\ref{thm:tci-query}), with core entries and arithmetic
carried out at $b$ bits of fixed-point precision. The position-diagonal
unitary $e^{i\beta\widetilde{V}_t}$ admits a quantum-circuit
implementation on $n_{\mathrm{sys}} = d \log_2 N$ system qubits,
with an $\mathcal{O}(b\,D_V)$-qubit arithmetic work register, using
\begin{equation}\label{eq:gate-count}
G_{\mathrm{step}}
\;=\;
\mathcal{O}\!\left(d\,N\,D_V^{\,2}\right)
\end{equation}
two-qubit gates up to factors polynomial in $b$, in a depth-$d$
sequence of per-site blocks. The count is a construction-level
estimate: it assumes $b$ is chosen so that fixed-point rounding
through the boundary-vector contraction is negligible, an error
analysis we do not perform.
\end{proposition}

\begin{proof}[Proof sketch]
The primary route is an arithmetic phase oracle that consumes the
$\widetilde{V}_t$ cores directly: sweeping the sites in order, a
work register holds the running $D_V$-dimensional boundary vector
in $b$-bit fixed point; at site $j$ the core matrix selected by the
position register, $A_j[x_j]$, is applied by
$\mathcal{O}(N D_V^{\,2})$ controlled multiply-adds; after site $d$
the scalar $\beta\widetilde{V}_t(x)$ is kicked into a phase and the
register is uncomputed. This gives the
count~\eqref{eq:gate-count} in a depth-$d$ sequence of per-site
blocks.
\end{proof}

\paragraph{End-to-end resource bounds}
Composing the per-step bounds of
Propositions~\ref{thm:tci-query}--\ref{thm:tci-circuit} with the
Trotter-step count of~\cite{bib:layden2025} yields total costs
for preparing the qsample $\ket{\Psi_T}$.

\begin{corollary}[Total resource cost of qsample preparation]
\label{cor:total-cost}
Let $\Vt$ have bond dimension at most $D_V$ uniformly in
$t \in [0,T]$. Let
$r = \mathcal{O}(d^{\,6} T^{2+2/s} / \epsilon^{1+2/s})$ be the
Trotter-step count of~\cite{bib:layden2025} required to prepare
$\ket{\Psi_T}$ to error $\epsilon$, and let
$\eta_V = \mathcal{O}(\epsilon / r)$ be the per-step compression
tolerance, which suffices because operator-norm errors of the $r$
unitary factors accumulate at most additively. Then
the total number of classical oracle queries to $\Vt$ in the
TCI build of all $r$ generator MPOs, and the total gate count of
the precompiled (Trotterized, time-ordered) qsample preparation
unitary $U(T)$, both satisfy
\begin{equation}\label{eq:total-cost}
Q_{\mathrm{total}},\;\; G_{\mathrm{total}}
\;=\;
\mathcal{O}\!\left(
\frac{d^{\,7}\, T^{2+2/s}\, N\, D_V^{\,2}}
{\epsilon^{1+2/s}}
\right),
\end{equation}
up to the precision and quasi-optimality factors of
Propositions~\ref{thm:tci-query}--\ref{thm:tci-circuit}. Both
quantities are polynomial in $d$ whenever
$D_V$ is polynomial in $d$.
\end{corollary}

\begin{proof}
Direct composition. $Q_{\mathrm{total}}$ is $r$ applications of
the per-step bound~\eqref{eq:query-perstep};
$G_{\mathrm{total}}$ is $r$ V-step circuits of cost
$G_{\mathrm{step}}$ from~\eqref{eq:gate-count} plus $r$ kinetic
substeps implemented by the quantum Fourier transform at cost
$\mathcal{O}(d (\log N)^2)$ each, and the V-step term dominates
whenever $N D_V^{\,2} \gtrsim d (\log N)^2$.
\end{proof}

The corollary shows that MPS-compressibility of $\Vt$ suffices to
make the per-step V-step cost polynomial in the black-box
(pointwise-evaluation) access model introduced at the start of this
section, where the worst case is $\Theta(N^d)$. For the classical
simulation as a whole one further condition is required: the
evolved state must itself remain representable at moderate bond
dimension $D_\Psi$, an empirical property of the flow
(Section~\ref{sec:results-bypass}) rather than a consequence of
$D_V$.

\paragraph{Specialization to orthogonal-mode Gaussian mixtures}
We can now compute~\eqref{eq:total-cost} for the principal
target class of the empirical study.

\begin{corollary}[Polynomial precompilation cost for
orthogonal-mode Gaussian mixture targets]\label{cor:gmm-precompile}
For the orthogonal-mode Gaussian mixture target of Section~\ref{app:a5},
assume $D_V = \mathcal{O}(d)$ uniformly in $t \in [\tau,T]$ for
some $\tau > 0$---an empirical assumption, consistent with the
achieved bond dimensions in our runs (Section~\ref{app:a5},
Section~\ref{sec:results-bypass}).
Substituting into~\eqref{eq:total-cost},
\begin{equation}\label{eq:gmm-precompile}
Q_{\mathrm{total}},\;\; G_{\mathrm{total}}
\;=\;
\mathcal{O}\!\left(
\frac{d^{\,9}\, N\, T^{2+2/s}}
{\epsilon^{1+2/s}}
\right),
\end{equation}
polynomial in $d$: degree~$9$ in the explicit $d$ factors with the grid
$N$ held symbolic, and one power higher once the grid bound
$N = \mathcal{O}(Ld(T/\epsilon)^{1/(2s)})$ of
Ref.~\cite{bib:layden2025} is substituted. The ratio of this cost to the
dense-grid-tabulation cost $\Theta(N^d)$ is exponentially small in
$d$ for any $N \ge 2$.
\end{corollary}

\begin{proof}
Substitution of $D_V^{\,2} = \mathcal{O}(d^2)$
into~\eqref{eq:total-cost} produces the stated bound.
\end{proof}

The restriction to $[\tau,T]$ excludes the $t \to 0$ singularity of
the analytic potential~\eqref{eq:Vt-gmm}: its $1/t$ prefactor is a
scalar and leaves the bond dimension unchanged, but it drives the
norm constants of Ref.~\cite{bib:layden2025} (and hence the
prefactors suppressed in~\eqref{eq:gmm-precompile}) to infinity as
$t \to 0$. We therefore start the flow at $t = 10^{-3}$.

\paragraph{Discussion}
The following remarks place Propositions~\ref{thm:tci-query} and~\ref{thm:tci-circuit} in context.

\begin{remark}[Generator versus exponential: the classical advantage
of TDVP]\label{rem:tdvp-vs-tci}
Both propositions charge every cost at the generator's bond
dimension $D_V$; a build-then-apply route that instead compresses the
unitary $e^{i\beta\Vt}$ into a global MPO incurs the exponential's bond
dimension $D_V^{\exp} \ge D_V$, which generically inflates with
$\beta\|\Vt\|$ (Methods, Section~\ref{app:b2}). TDVP realizes the
generator-route accounting classically: it integrates
$i\partial_s\ket{\Psi} = \Vt\ket{\Psi}$ under the $\Vt$-MPO, never
materializing the exponential
(Methods, Sections~\ref{app:b4}--\ref{app:b5}), so the query
bound~\eqref{eq:query-perstep} applies to it without modification.
\end{remark}

\begin{remark}[State-loading depth]\label{rem:state-loading}
Loading the final state $\ket{\Psi_T}$, of MPS bond dimension
$D_\Psi$, is a separate primitive from the generator costs charged in
Propositions~\ref{thm:tci-query}--\ref{thm:tci-circuit}. We review its depth
ladder (linear, logarithmic, or constant, by construction) in the
Introduction, under explicit gate,
precision, and connectivity assumptions; we do not compile it here.
\end{remark}

\begin{remark}[Smoothness alone is not enough]
\label{rem:beyond-smoothness}
Theorem~1 of~\cite{bib:layden2025} controls the
spatial-discretization error through Sobolev-type norms of
$\sqrt{p_t}$ and $\Vt\sqrt{p_t}$; such smoothness does not by
itself imply decaying unfolding singular values or a small bond
dimension $D_V$ of the kind entering Corollary~\ref{cor:total-cost}.
The bond dimension instead depends on the geometry of the target,
in particular, the number and spatial arrangement of its modes. The correlated Gaussian mixture regime
illustrates the gap: a smooth target with genuine cross-axis
covariance drives the required bond dimension well above that of the
orthogonal case, and at high $d$, this is the binding constraint.
Identifying which combinations of mode geometry, smoothness, and
support guarantee $D_V$ polynomial in $d$ is left to future work
and is the starting point for any rigorous analysis of MPS+TCI
wavefunction flow simulation beyond the regime studied here.
\end{remark}

\subsection{Rare-event indicator oracle}\label{app:rare-oracle}
The amplification results of Section~\ref{sec:results-rare}
require a marking
oracle $O_f\ket{x}=(-1)^{f(x)}\ket{x}$ flagging the rare-event subspace, with
\[
  f(x)=\mathbb{1}\!\left[\min_{j}\lVert x-c_j\rVert^2>(k\sigma)^2\right]
       \cdot
       \mathbb{1}\!\left[\min_{j}\lVert x-c_j\rVert < \lVert x-x_0\rVert\right],
\]
$\{c_j\}_{j=1}^{2d}$ the mode centers, $x_0$ the domain center, and $k\sigma$
the tail radius. The second factor encodes, coherently, the same exclusion
that the classical estimator applies by post-selection: it marks $x$ only when $x$
is closer to some mode than to the domain center, so untransported source mass
sitting near $x_0$ (which is also $>\!k\sigma$ from every mode) is not
counted as rare. Without it, the coherent oracle and the post-selected
classical event would integrate different regions.

The oracle is an efficient reversible circuit. With $x$ held in
$d$ coordinate registers, $O_f$ forms the $2d$ squared distances
$\lVert x-c_j\rVert^2=\sum_{i=1}^{d}(x_i-c_{j,i})^2$, minimizes them, and
compares both to $(k\sigma)^2$ and to $\lVert x-x_0\rVert^2$:
$\mathcal{O}(d^2)$ reversible additions,
multiplications, and comparisons on $\mathcal{O}(d\log N)$ qubits---polynomial
in $d$ and independent of the $N^d$ grid. The query complexity advantage of
Section~\ref{sec:results-rare} ($\mathcal{O}(1/\sqrt{\prare})$ vs.\ $\mathcal{O}(1/\prare)$ for sampling)
therefore carries over to end-to-end
gate cost, because each query invokes only this $\mathrm{poly}(d)$ oracle
alongside $U_{\mathrm{prep}}$.

We estimate the centers and
scale the oracle needs from the common (bulk)
samples: a
handful of draws per mode ($\gtrsim 10$) fixes each cluster's location and
covariance to high accuracy, so no rare events are needed to construct $O_f$. The
expensive quantity, the tail mass, is the one the quantum routine
accelerates; specifying the region it integrates over is cheap classical
preprocessing on abundant data. This presumes the rare region is known in
advance, which holds here because the tail is defined by a distance rule
against modes the bulk already reveals. When the rare events are not
known ahead of time, they cannot be flagged for amplification at all, and
discovery becomes part of the problem; Guo \emph{et al.}~\cite{bib:guo2026rare}
give a quantum algorithm for that setting, attaining the optimal scaling in
the rarity threshold without first learning which events are rare. For a general (non-Gaussian) target, the tail
is a threshold on the learned log-density or potential,
$f(x)=\mathbb{1}[V_T(x)>\tau]$, and the same tensor network representation of
$\Vt$ that drives the flow furnishes the indicator, so the oracle need not
be hand-built.

Two caveats apply. First, we simulate the exact
Grover outcome statistics for amplification rather than compiling $O_f$ and
$U_{\mathrm{prep}}$ to gates; the accounting just given is analytic, not measured.
Second, and more fundamentally, the query complexity advantage is generic to
amplitude estimation and amplification, and does \emph{not} by itself establish a dequantization-proof
separation. Whenever the observable $O_f$ and the prepared state are both
low-rank, the mass $\langle\Psi_T|O_f|\Psi_T\rangle$ can be evaluated by
classical tensor contraction, and no quantum advantage remains. The
orthogonal-mode tail indicator used here is plausibly of moderate
bond dimension, so this demonstration should be read as an illustration
of the downstream routine on a faithfully prepared state, not as a claim of
dequantization-resistant advantage on this particular target. The advantage applies
to observables that are efficiently oracle-accessible yet not efficiently
contractible---e.g.\ verifier rewards or acceptance predicates defined by an
external circuit rather than a low-rank field.

\section{Supplementary results}\label{app:supp-results}

\subsection{Cost--accuracy Pareto frontiers}\label{app:supp-pareto}
In Figure~\ref{fig:supp-pareto} we show the cost--accuracy Pareto
frontiers underlying the per-$d$ trajectory of
Figure~\ref{fig:cost}. We arrange the figure as a $2 \times 3$
grid: each column corresponds to a dimensionality
$d \in \{3, 4, 5\}$, the top row plots SW against the
dimensionless memory ratio (MPS / Dense-grid), and the bottom row
plots SW against wall-clock time per Trotter sweep. Each point is a
hyperparameter cell from the search grid of
Table~\ref{tab:search-grid} of the main text; the staircases are per-method
Pareto fronts; dashed and dotted lines mark the Dense-grid and JAM
reference SW, respectively. Across every column, JAM and the
TDVP variants form the lower-left Pareto frontier, and the
MPS-to-Dense-grid advantage on both
axes widens as $d$ grows. We restrict the supplementary Pareto
sweep to $d \le 5$, where each cell has $\sim 30$
hyperparameter samples per method; $d \ge 6$ was sampled too
sparsely for a meaningful per-method Pareto and is represented
in the main text only by the best-cell trajectory of
Figure~\ref{fig:cost}.

\begin{figure}[tbp]
\centering
\includegraphics[width=0.95\textwidth]{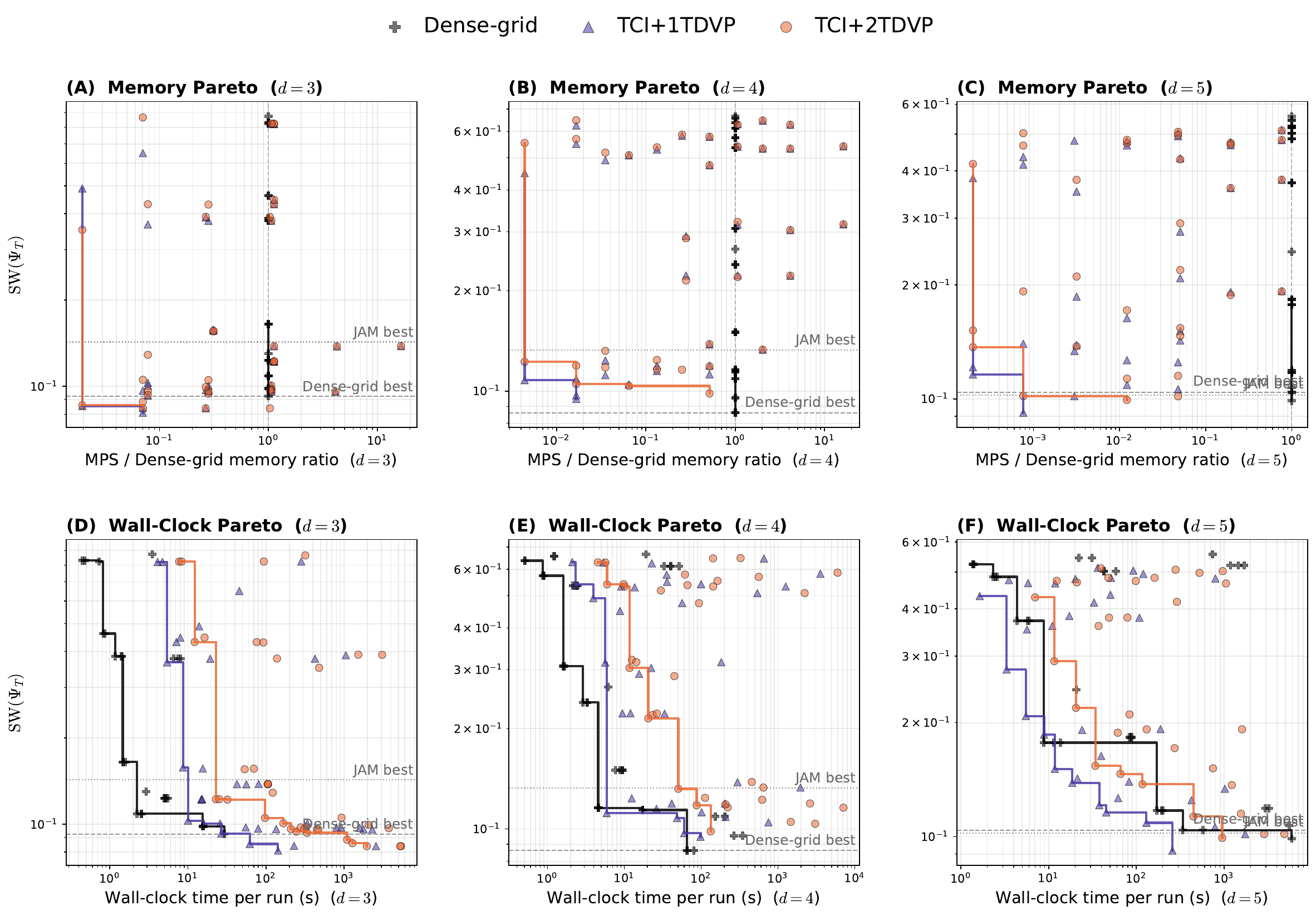}
\caption{\textbf{Cost--accuracy Pareto frontiers across spatial
dimensions.} Top row (panels a--c): SW vs MPS/Dense-grid memory ratio
at $d = 3$, $4$, $5$. Bottom row (panels d--f): SW vs evolution wall-clock time (s) at $d = 3$, $4$, $5$ ---
summed per Trotter step, the same quantity as Figure~\ref{fig:cost}(b), so the two are directly comparable. Each point is a
hyperparameter cell; staircases are per-method Pareto fronts;
dashed lines mark the Dense-grid reference SW and dotted lines the
JAM reference. JAM and the TDVP variants occupy the lower-left
frontier in every panel. The
MPS-to-Dense-grid advantage on both axes widens as $d$ increases.}
\label{fig:supp-pareto}
\end{figure}

\subsection{Replicated best-cell audit}\label{app:audit}
To put the best-cell numbers of Table~\ref{tab:sw} on a replicated, two-metric
footing, we re-ran the audit with $n{=}10$ random seeds per cell, reporting
sliced-Wasserstein as mean${\pm}$std and (in parentheses) the seed-mean
radial basis function (RBF) kernel MMD as a point value. Rather than an exhaustive $(N,K,D)$ grid we
locate the optimum
by a theory-informed coordinate descent, justified by how each axis behaves: the spatial error is monotone in the
grid $N$ ($\epsilon_{\mathrm{spatial}}\!\propto\!(Ld/N)^{2s}$,
Ref.~\cite{bib:layden2025}), so we fix
$N$ at the largest feasible grid. The Trotter count $K$ and the bond dimension
$D$ each show a nonmonotonic (interior) optimum in our sweeps, so we
sweep each and take the best $K^*,D^*$. We attribute the nonmonotonicity to
interactions among Trotter splitting, finite-step (Dirac--Frenkel)
integration, and the truncation rule---accumulated per-step error growing with
$K$ against the shrinking Trotter error, and truncation interacting with the
larger variational space at bigger $D$---rather than to a single mechanism. The TDVP routines use a single Dirac--Frenkel sweep per V-step;
over-iterating it degrades accuracy and incurs ${\sim}4{\times}$ higher cost. This
audit uses a \emph{trained} velocity potential and includes only JAM and
TCI+1TDVP, so it complements rather than replicates Table~\ref{tab:sw} (which
uses the analytic $V$ and also reports Dense-grid and TCI+2TDVP); its absolute
scale differs because it is evaluated at a smaller sample size. The
relevant results are the TCI+1TDVP-vs-JAM ordering and the
location of the optima.

\begin{table}[!htbp]
\centering
\caption{\textbf{Replicated $n{=}10$ best-cell audit} (trained velocity potential),
sliced-Wasserstein as mean${\pm}$std over the $10$ seeds, with the seed-mean MMD
(a point value, no spread) in parentheses; lower is better. The
best cell $(N,K^*,D^*)$ per entry is located by coordinate descent; the located
corners are interior ($K^*{\sim}40$, $D^*{\sim}8$), not at the box maxima.}
\label{tab:audit}
\begin{tabular}{lcccc}
\toprule
method & $d{=}2$ & $d{=}3$ & $d{=}4$ & $d{=}5$ \\
\midrule
Target--target (finite-sample) & $0.023{\pm}0.001$ & $0.025{\pm}0.002$ & $0.023{\pm}0.001$ & $0.024{\pm}0.001$ \\
JAM & $0.111{\pm}0.006$ & $0.084{\pm}0.005$ & $0.073{\pm}0.004$ & $0.077{\pm}0.004$ \\
TCI+1TDVP & $0.206{\pm}0.007$ (0.051) & $0.142{\pm}0.006$ (0.066) & $0.146{\pm}0.007$ (0.095) & $0.156{\pm}0.005$ (0.065) \\
\bottomrule
\end{tabular}
\end{table}

Two findings supplement the main text. \emph{(i)~TCI+1TDVP is the robust,
Pareto-optimal method.} Its SW stays in a narrow band ($0.14$--$0.21$)
across $d{=}2$--$5$
(Table~\ref{tab:audit}) at small $D^*$ and interior $K^*$. Beyond the
tabulated range, at the $d{\in}\{6,7,8\}$, $D{=}64$ frontier (the trained MPS
$V$-step, not in Table~\ref{tab:audit}; with no runtime TCI, the integrator is 1TDVP proper rather than TCI+1TDVP), 1TDVP remains the only
method that is both accurate and tractable. 1TDVP completes in tens of
seconds, consistent with the trained MPS potential eliminating the runtime TCI construction, and hence far below the analytic-$V$ runtime-TCI wall-clock times reported in
Section~\ref{sec:discussion}. TCI+2TDVP with runtime TCI, by contrast, exhausts
a 32\,GB worker, and a
128\,GB retry at $d{=}6$ ran over two hours without finishing. This behavior
underlies the cost scaling claim and motivates 1TDVP+SE
(Methods, Section~\ref{app:se-adaptive}). \emph{(ii)~The optimum is a non-obvious corner.} Best
cells sit at large $N$ but small $D^*{\sim}4$--$16$ and interior
$K^*{\sim}20$--$80$, not at the largest $K,D$; accuracy degrades on either side
of these corners, the nonmonotonicity described earlier in this subsection. MMD
tracks SW throughout, so the TCI+1TDVP-vs-JAM comparison holds under an
independent metric.

\subsection{Mode sharpness and grid resolution}\label{app:supp-sharpness}
The smoothness conditions entering the error theory
(Section~\ref{app:a7}) predict that the
difficulty of a wavefunction flow is set by how sharp the target's modes
are relative to the spatial grid, the dimensionless ratio
$\sigma/\mathrm{d}x = \sigma N / L$. We probe this ratio along two
independent routes on the orthogonal Gaussian mixture (analytic $V_t$, $d{=}8$),
measuring the instantaneous deviation of the flow from the
probability flow action matching reference,
$\mathrm{SW}^{\mathrm{wf}}(\Psi_t) - \mathrm{SW}^{\mathrm{pf}}(t)$.

In Figure~\ref{fig:panelA-sigma} we vary the mode width $\sigma$ at fixed
grid $N{=}32$ (the sharpness route). The deviation grows steadily as the
modes sharpen: at $\sigma{=}1.0$ ($\sigma/\mathrm{d}x{=}4.0$) the flow
tracks the probability flow reference almost exactly, whereas at $\sigma{=}0.35$
($\sigma/\mathrm{d}x{=}1.4$) it departs by up to ${\sim}0.15$ near $t{=}1$.
At every $\sigma$, 2TDVP stays closer to the reference than 1TDVP+SE at
matched $K$, and a larger $K$ helps mainly the sharp-mode panels.

In Figure~\ref{fig:panelN-resolution} we instead fix $\sigma{=}0.5$ and
vary the grid resolution $N$ (the resolution route), reaching the same set
of $\sigma/\mathrm{d}x$ ratios by shrinking $\mathrm{d}x$ rather than
widening the mode. The two routes do not collapse onto a common
curve: along the resolution route the deviation grows with $N$,
opposite to the sharpness route at matched $\sigma/\mathrm{d}x$. Refining
the grid therefore does not by itself ease the flow; it also raises the
Trotter constant (which scales as $d^2 N^4$~\cite{bib:layden2025}), so a
finer grid demands proportionally more Trotter steps. The ratio
$\sigma/\mathrm{d}x$ captures the resolution of the modes but is not, on
its own, a sufficient statistic for the flow error.

\begin{figure}[tbp]
\centering
\includegraphics[width=\textwidth]{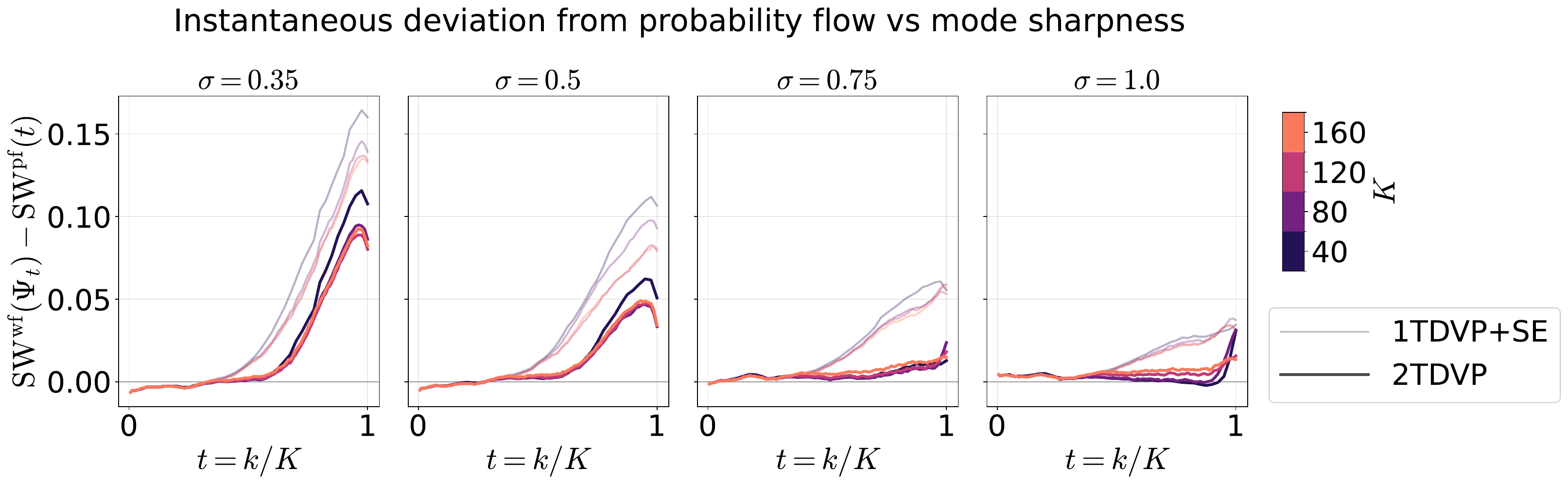}
\caption{\textbf{Sharpness route: instantaneous deviation from the
probability flow versus mode width.} One panel per mode width
$\sigma\in\{0.35,0.5,0.75,1.0\}$ at fixed $N{=}32$ (orthogonal Gaussian mixture,
analytic $V_t$, $d{=}8$); color encodes the Trotter count $K$, with
1TDVP+SE drawn light and 2TDVP dark. The deviation
$\mathrm{SW}^{\mathrm{wf}}(\Psi_t)-\mathrm{SW}^{\mathrm{pf}}(t)$ grows as the modes
sharpen (smaller $\sigma$, smaller $\sigma/\mathrm{d}x$).}
\label{fig:panelA-sigma}
\end{figure}

\begin{figure}[tbp]
\centering
\includegraphics[width=\textwidth]{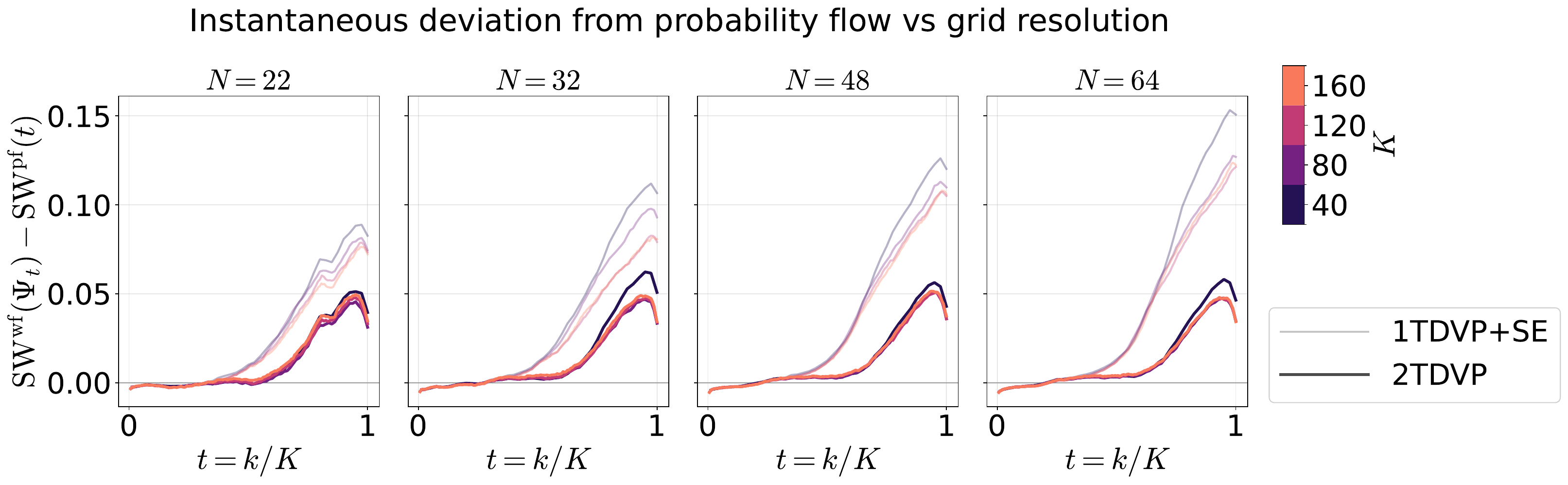}
\caption{\textbf{Resolution route: instantaneous deviation versus grid
resolution.} As in Figure~\ref{fig:panelA-sigma} but at fixed
$\sigma{=}0.5$ with the grid $N\in\{22,32,48,64\}$ varied, reaching the
same $\sigma/\mathrm{d}x$ ratios by refining $\mathrm{d}x$. The deviation
grows with $N$, opposite to the sharpness route at matched
$\sigma/\mathrm{d}x$: the two routes do not collapse because a finer grid
also raises the Trotter constant.}
\label{fig:panelN-resolution}
\end{figure}

\subsection{Cost across Trotter count and mode sharpness}\label{app:supp-ksigma}
In Figure~\ref{fig:panelD-ksigma} we summarize the cost--accuracy tradeoff
across both parameters simultaneously: the total deviation
$\int_0^1 |\mathrm{SW}^{\mathrm{wf}}(\Psi_t) - \mathrm{SW}^{\mathrm{pf}}(t)|\,\mathrm{d}t$
against measured wall-clock time, with color encoding the Trotter count $K$,
marker shape the mode width $\sigma$, and fill the $V$-step (orthogonal
Gaussian mixture, analytic $V_t$, $d{=}8$, $N{=}32$; each $\sigma$ scored against its
own probability flow floor). Sharper modes (smaller $\sigma$) sit higher at fixed
cost, echoing the sharpness route of Section~\ref{app:supp-sharpness}; 2TDVP
reaches lower deviation
than 1TDVP+SE but at roughly an order of magnitude more wall-clock time,
tracing the same cost--accuracy frontier reported in the main text.

\begin{figure}[tbp]
\centering
\includegraphics[width=0.72\textwidth]{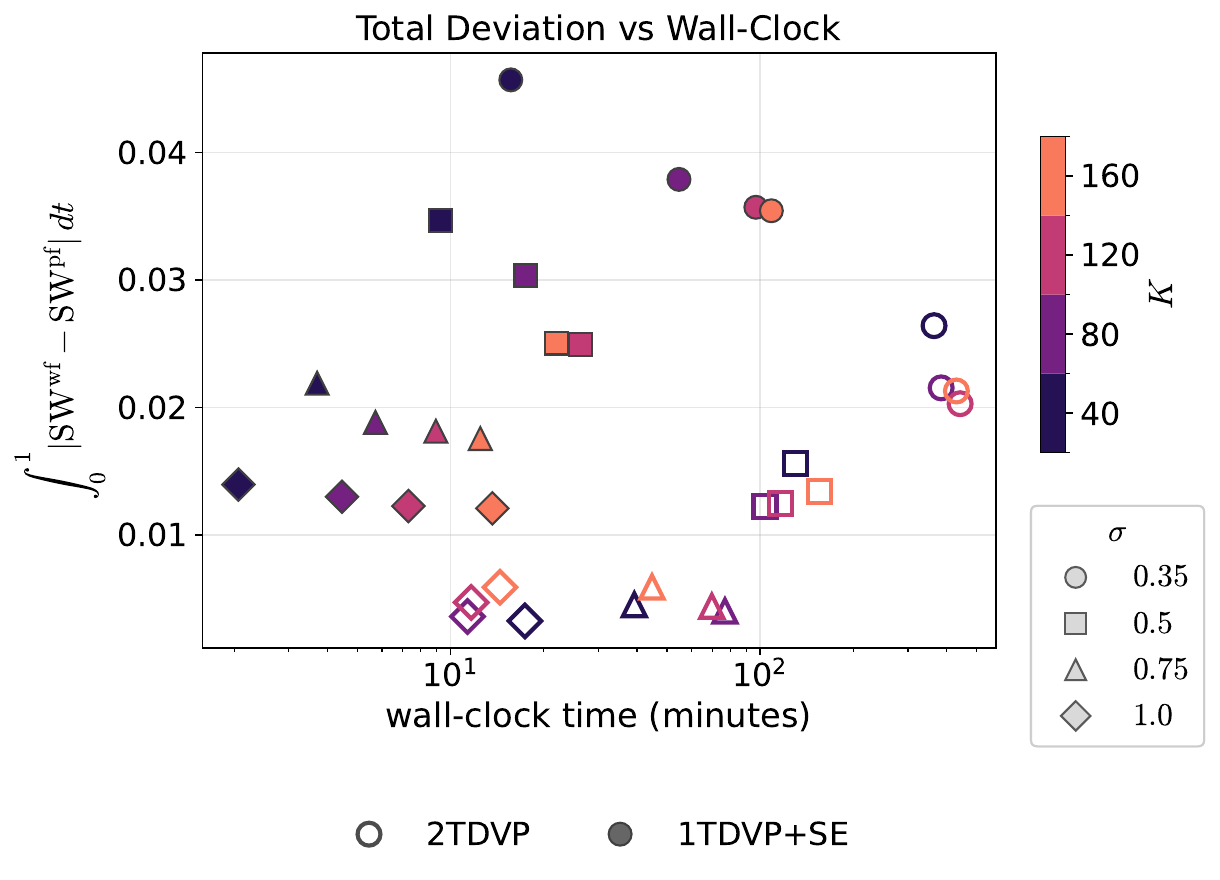}
\caption{\textbf{Cost--accuracy across Trotter count and mode sharpness.}
Total deviation versus measured wall-clock time (orthogonal Gaussian mixture,
analytic $V_t$, $d{=}8$, $N{=}32$). Color encodes the Trotter count $K$,
marker shape the mode width $\sigma$, and fill the $V$-step (open: 2TDVP;
filled: 1TDVP+SE). Lower-left is better; sharper modes and larger $K$ both
move points up and to the right.}
\label{fig:panelD-ksigma}
\end{figure}

\subsection{Analytic versus trained velocity potential}\label{app:supp-analytic-trained}
In Figure~\ref{fig:analytic-trained} we repeat the $\sigma{=}0.5$
$K$-sweep with the two velocity sources side by side: the analytic Gaussian mixture
potential (left) and the pre-trained MPS potential (right). The top row shows the
instantaneous deviation from the probability flow, and the bottom row shows the
total deviation against wall-clock time. The trained oracle reproduces
the analytic-$V$ behavior (with a small deviation that remains controlled until
late $t$ and shrinks with $K$), confirming that the learned cores faithfully drive
the flow. The left column carries only the compression and
integration errors because the velocity is exact; the right column adds
the velocity learning error, and the close agreement between the columns
indicates that this contribution is small. One caveat on the comparison: the two
oracles are not the same function class on the torus. Both trained oracles
are periodic by construction (the MLP through its $\sin/\cos$ features at
period $L$, the trained MPS potential through its modular grid indexing). In contrast, the
analytic $\lVert x\rVert^2/(2t)$ is not, so $\nabla V$ reverses sign at
the seam. With modes at $\pm3$, $\sigma=0.5$, and the wall at $\pm4$,
roughly $2\%$ of target mass sits outside the box and wraps, so a residual
difference between the columns is expected on those grounds alone and
should not be read entirely as oracle quality.

\begin{figure}[tbp]
\centering
\includegraphics[width=0.95\textwidth]{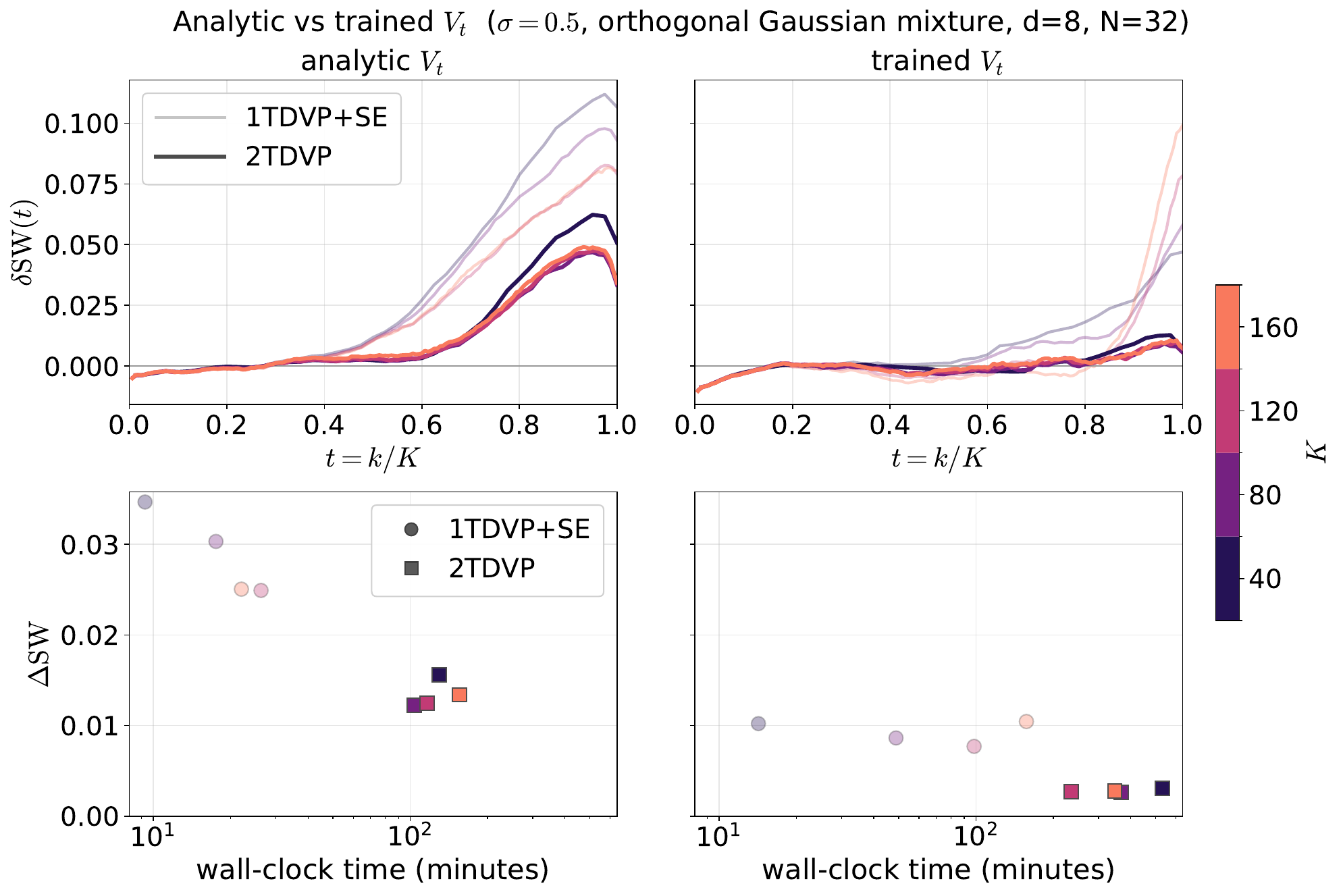}
\caption{\textbf{Analytic versus trained $V_t$ at $\sigma{=}0.5$.}
Orthogonal Gaussian mixture, $d{=}8$, $N{=}32$. Columns: analytic Gaussian mixture velocity (left)
and pre-trained MPS potential (right). Top row: instantaneous deviation from the probability flow,
$\delta\mathrm{SW}(t) = \mathrm{SW}^{\mathrm{wf}}(\psit)-\mathrm{SW}^{\mathrm{pf}}(p_t)$,
where $\mathrm{SW}(q)$ is the distance from samples of the source $q$ to samples of
the target ($\psit$ by Born measurement, the marginal density $p_t$ by transport),
with $K$ by color. Bottom row: the deviation accumulated
over the flow, $\Delta\mathrm{SW} = \int_0^1|\delta\mathrm{SW}|\,dt$, versus
wall-clock time (1TDVP+SE and 2TDVP). The trained oracle reproduces the analytic-$V$ flow,
isolating velocity learning error as small.}
\label{fig:analytic-trained}
\end{figure}

\FloatBarrier

\bibliography{sn-bibliography}

\end{document}